\documentclass[12pt]{article}

\usepackage{a4wide}

\usepackage{doi}
\usepackage{amsfonts, amsmath, amssymb}
\usepackage{color}
\usepackage{tikz,url}
\usetikzlibrary{calc,shapes}

\newcommand{\eq}{\leftrightarrow}
\newcommand{\Eq}{\Leftrightarrow}
\newcommand{\imp}{\rightarrow}
\newcommand{\Imp}{\Rightarrow}

\newcommand{\et}{\wedge}
\newcommand{\vel}{\vee}
\newcommand{\Et}{\bigwedge}

\renewcommand{\phi}{\varphi}
\newcommand{\union}{\cup}
\newcommand{\Union}{\bigcup}
\newcommand{\inter}{\cap}
\newcommand{\Inter}{\bigcap}

\newcommand{\power}{\mathcal P}

\newcommand{\bisim}{{\raisebox{.3ex}[0mm][0mm]{\ensuremath{\medspace \underline{\! \leftrightarrow\!}\medspace}}}}

\newcommand{\domain}{\mathcal{D}}

\usepackage[thmmarks]{ntheorem}
\theorembodyfont{\rmfamily}
\theoremsymbol{\ensuremath{\dashv}}
\usepackage{newproof}

\newtheorem{theorem}{Theorem}
\newtheorem{example}[theorem]{Example}
\newtheorem{definition}[theorem]{Definition}
\newtheorem{proposition}[theorem]{Proposition}

\newtheorem{corollary}[theorem]{Corollary}

\newcommand{\lang}{\mathcal L}

\newcommand{\C}{\mathcal C}
\newcommand{\VV}{\mathcal V}
\newcommand{\FF}{\mathcal F}
\newcommand{\weg}[1]{}

\newcommand{\powerset}{\mathcal{P}}
\newcommand{\RR}{\mathfrak R}

\newcommand{\cg}{R}

\newcommand{\cp}{\pmb{R}}

\newcommand{\p}{p}

\newcommand{\np}{\overline{p}}

\title{Communication Pattern Logic: \\ Epistemic and Topological Views}

\author{Armando Casta\~{n}eda \and Hans van Ditmarsch \and David A.\ Rosenblueth \and Diego A.\ Vel\'azquez\thanks{Armando, David, and Diego are affiliated to UNAM in Mexico. Hans is affiliated to University of Toulouse, CNRS, IRIT in France.}}
\date{}

\begin{document}

\maketitle

\begin{abstract}
We propose \emph{communication pattern logic}. A communication pattern describes how processes or agents inform each other, independently of the information content. The full-information protocol in distributed computing is the special case wherein all agents inform each other. We study this protocol in distributed computing models where communication might fail: an agent is certain about the messages it receives, but it may be uncertain about the messages other agents have received. In a dynamic epistemic logic with distributed knowledge and with modalities for communication patterns, the latter are interpreted by updating Kripke models. We propose an axiomatization of communication pattern logic, 
and we show that collective bisimilarity (comparing models on their distributed knowledge) is preserved when updating models with communication patterns. We can also interpret communication patterns by updating \emph{simplicial complexes}, a well-known topological framework for distributed computing. We show that the different semantics correspond, and propose collective bisimulation between simplicial complexes.
\end{abstract}


\section{Introduction} \label{sec.intro}

Epistemic logic \cite{hintikka:1962} investigates knowledge and change of knowledge in multi-agent systems, both in temporal epistemic logics~\cite{halpernmoses:1990,Pnueli77,dixonetal.handbook:2015} and in dynamic epistemic logics~\cite{hvdetal.del:2007}, including synchronous and asynchronous semantics based on action histories \cite{jfaketal.JPL:2009,degremontetal:2011,BalbianiDG22}. Recently, logics of knowledge and its dynamics have been proposed for simplicial complexes, topological structures modelling (a)synchronous computation \cite{herlihyetal:2013,GoubaultLR21,ledent:2019,goubaultetal_postdali:2021}. 
A state (world) in a multi-agent Kripke model describes the knowledge of all agents in an integrated way. In contrast, a state (local state, vertex) in a simplicial complex describes the knowledge of a single agent. To obtain a complete state description we must then combine the different agent views. 
Logics with a topological interpretation to study distributed systems are appealing, given the  connection between 
combinatorial topology and distributed computing \cite{BG93,HS93,SZ93,herlihyetal:2013}. Epistemic logic with such a topological interpretation can enrich epistemic analysis with the topological machinery to
study distributed computing problems, such as formalizing topological subdivisions as epistemic updates, formalizing intersection in higher dimensions with distributed knowledge, or formalizing connectivity, or preservation of connectivity after update, with common knowledge \cite{GoubaultLR21,ledent:2019,diego:2021,goubaultetal_postdali:2021,hvdetal.simpl:2022}.

In temporal and dynamic epistemic logics we typically model how specific information contained in a message from or to an agent changes the global state of information. In contrast, the \emph{communication patterns} of this contribution describe how agents communicate with each other independently of the information content of the message. Precursor dynamic epistemic logics of such information exchange are \cite{AgotnesW17,Baltag20,diego:2021}. In \cite{AgotnesW17} the authors model a subgroup $B$ of the set of agents $A$ sharing all their knowledge. As intuitively this describes how to make the distributed knowledge between the agents in $B$ common knowledge between the agents in $B$ (modulo Moorean phenomena \cite{moore:1942,Roelofsen07,hollidayetal:2010}), this is called \emph{resolving distributed knowledge}. The agents not in $B$ are aware that the agents in $B$ perform this update. In \cite{Baltag20} this idea is generalized to describe agents sharing their knowledge in arbitrary ways, as when only agent $a$ gets to know all that $b$ knows whereas only $b$ gets to know all that $c$ knows. The authors call such information exchanges \emph{reading events}. The communication patterns we investigate originate in \cite{diego:2021}, which models that agents may be \emph{uncertain} about what other agents learn in such updates, such as when $a$ gets to know all that $b$ knows, while $a$ remains uncertain whether $b$ also learns all that $a$ knows. Before we proceed, let us look in some detail at a simple communication pattern, that also counts as resolving distributed knowledge.

\begin{example} \label{exampleone}
Suppose that Anne knows whether $p_a$ whereas Bill knows whether $p_b$, as depicted in Fig.~\ref{annebill}. A simple communication pattern consists of Anne and Bill successfully telling each other all they know. We can also see this as non-deterministic choice between four public announcements: Anne telling Bill that she knows $p_a$ and Bill telling Anne that he knows $p_b$, or Anne telling Bill that she knows $\neg p_a$ and Bill telling Anne that he knows $\neg p_b$, and so on. That makes four, mutually exclusive, public announcements compared to a single communication pattern. We represent this update as the update of a Kripke model but also as a corresponding update of a simplicial complex. In Fig.~\ref{annebill}, states of the Kripke model and vertices of the simplicial complex are labelled with values of propositional variables $p_a$ and $p_b$, where $\np_a$ and $\np_b$ stand for $\neg p_a$ and $\neg p_b$.

A state in the Kripke model corresponds to an edge in the simplicial complex, and a link in the Kripke model (an equivalence class) corresponds to a vertex in the simplicial complex. The representations are perfectly dual and satisfy the same formulas. 

If there are three agents, a state in the Kripke model corresponds to a triangle in the simplicial complex (as in the other examples in our contribution), if there are four, to a tetrahedon, and so on. Formal explanations will only be given from Sect.~\ref{sec.structures} onwards.

\begin{figure}[h]
\center
\scalebox{.85}{
\begin{tabular}{ccc}
\begin{tikzpicture}
\node (00) at (0,0) {$\np_a\np_b(s)$};
\node (01) at (0,2) {$\np_a\p_b(t)$};
\node (10) at (2.5,0) {$\p_a\np_b(v)$};
\node (11) at (2.5,2) {$\p_a\p_b(u)$};
\draw (00) -- node[above,fill=white,inner sep=1pt] {$b$} (10);
\draw (01) -- node[above,fill=white,inner sep=1pt] {$b$} (11);
\draw (00) -- node[right,fill=white,inner sep=1pt] {$a$} (01);
\draw (10) -- node[right,fill=white,inner sep=1pt] {$a$} (11);
\end{tikzpicture}
%
&
\begin{tikzpicture}
\node (t) at (4,1) {$\Imp$};
\node (t) at (4,0) {\color{white}$\Imp$};
\end{tikzpicture}
&
%
\begin{tikzpicture}
\node (00tt) at (5,0) {$\np_a\np_b(s)$};
\node (01tt) at (5,2) {$\np_a\p_b(t)$};
\node (10tt) at (7.5,0) {$\p_a\np_b(v)$};
\node (11tt) at (7.5,2) {$\p_a\p_b(u)$};
%
\end{tikzpicture}
 \\
\begin{tikzpicture}
\node (l) at (-1.5,0) {$\np_a$};
\node (r) at (1.5,0) {$\p_a$};
\node (d) at (0,-1.5) {$\np_b$};
\node (u) at (0,1.5) {$\p_b$};
\draw (l) -- node[above,fill=white,inner sep=1pt] {$t$} (u);
\draw (l) -- node[above,fill=white,inner sep=1pt] {$s$} (d);
\draw (r) -- node[right,fill=white,inner sep=1pt]  {$u$} (u);
\draw (r) -- node[right,fill=white,inner sep=1pt] {$v$} (d);
\end{tikzpicture}
&
\begin{tikzpicture}
\node (t) at (4,1) {$\Imp$};
\node (t) at (4,-.5) {\color{white}$\Imp$};
\end{tikzpicture}
&
\scalebox{.8}{
\begin{tikzpicture}
\node (lt) at (3.5,0) {$\np_a$};
\node (rt) at (7.5,0) {$\p_a$};
\node (blt) at (3.5,-1) {$\np_a$};
\node (brt) at (7.5,-1) {$\p_a$};
\node (dt) at (5,-2.5) {$\np_b$};
\node (ut) at (5,1.5) {$p_b$};
\node (dtt) at (6,-2.5) {$\np_b$};
\node (utt) at (6,1.5) {$p_b$};
\draw (lt) -- node[above,fill=white,inner sep=1pt] {$t$} (ut);
\draw (blt) -- node[above,fill=white,inner sep=1pt] {$s$} (dt);
\draw (rt) -- node[right,fill=white,inner sep=1pt]  {$u$} (utt);
\draw (brt) -- node[right,fill=white,inner sep=1pt] {$v$} (dtt);
\end{tikzpicture}
}
\end{tabular}
}
\caption{Anne and Bill tell each other all they know}
\label{annebill}
\end{figure}
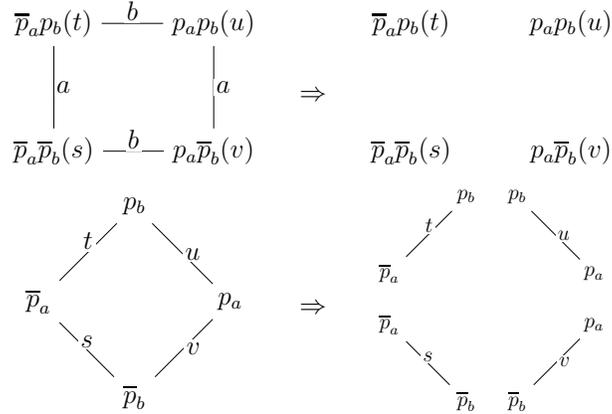
\end{example}

Communication pattern logic has epistemic modalities for distributed knowledge and dynamic modalities for communication patterns, that are interpreted by updating models. We give an axiomatization that is a slight generalization of those in \cite{AgotnesW17,Baltag20}, and we show that communication patterns preserve collective bisimilarity. 
We also interpret communication pattern logic in a topological semantics for simplicial complexes, where we provide correspondence between the Kripke and the simplicial semantics, and we additionally propose collective bisimulation between simplicial complexes, with the obvious sanity checks. We finally prove that the logics are the same.

Communication patterns describe distributed systems known as \emph{oblivious}, 
where messages between agents can be lost in a round of communication, and where the same communication failures may occur at any round. The repeated application of a communication pattern describes the dynamics of information exchange through such rounds of communication. 
Oblivious distributed systems have been extensively studied in distributed computing (e.g.\ \cite{ABFGLS11,AG13,CBS09,NSW19}), including topological modellings of the 
\emph{iterated immediate snapshot} that captures the computational power of multicore architectures with only simple read/write operations for exchanging information \cite{herlihyetal:2013} (see also Ex.~\ref{snapshot}, later).

We conclude the introduction with an overview of the content of our contribution.
Sects.~\ref{sec.structures} and \ref{sec.logic} define the syntax and semantics of communication pattern logic, and provide its axiomatization. Sect.~\ref{sec.simp} interprets communication pattern logic on simplicial complexes. Sect.~\ref{sec.comparison} concludes with further research.

\section{Structures} \label{sec.structures}

Given are a finite set of \emph{agents} $A$ and a set of \emph{propositional variables} $P \subseteq P' \times A$, where $P'$ is a countable set. For $B \subseteq A$ and $Q \subseteq P$, $Q \inter (P' \times B)$ is denoted $Q_B$ (where  $Q_a$ is $Q_{\{a\}}$), and $(p,a) \in P$ is denoted $p_a$. The set $P_a$ consists of the \emph{local variables} of agent $a$. 

\begin{definition}[Epistemic model]
An {\em epistemic model} $M$ is a triple $(W,\sim,L)$, where for each $a \in A$, $\sim_a$ is an {\em equivalence relation} on the \emph{domain} $W$ (also denoted $\domain(M)$) consisting of \emph{states} (or \emph{worlds}), and where $L: W \imp \power(P)$ is the \emph{valuation} (function). Given $w \in W$, $(M,w)$ is a \emph{pointed epistemic model}. For $\Inter_{a \in B} \sim_a$ we write $\sim_B$, and for $\{ w' \in W \mid w' \sim_a w\}$ we write $[w]_a$. We further require epistemic models to be \emph{local}, that is: for each $a \in A$ and $v,w \in W$: if $v \sim_a w$ then $L(v)_a =L(w)_a$. 
\end{definition}
An epistemic model encodes uncertainty among the agents about the value of other agents' local variables and about the knowledge of other agents.

\begin{definition}[Communication graph]
A {\em communication graph} $\cg$ is a reflexive binary relation on the set of agents $A$, that is, $\cg \in \power(A \times A)$ and such that for all $a \in A$, $(a,a) \in \cg$. A {\em communication pattern} $\cp$ is a set of communication graphs, that is, $\cp \subseteq \power(A \times A)$. A {\em pointed communication pattern} is a pair $(\cp,\cg)$, where $\cg\in\cp$.
\end{definition}
Expression $(a,b) \in \cg$ means that the message sent by $a$ is received by $b$, or, more precisely, that there is a channel from $a$ to $b$ on which $a$ has transmitted all she knows. In terms of epistemic logic: an agent transmits or sends `all she knows', if she announces the distinguishing formula of her equivalence class.\footnote{Given an epistemic model with domain $S$ and $S' \subseteq S$, a distinguishing formula for $S'$ is true in all states in $S'$ and false in all states in the complement of $S'$.} Observe that what $a$ knows is different in each of $a$'s equivalence classes. Ex.~\ref{exampleone} demonstrates this well.

For $(a,b) \in \cg$ we write $a \cg b$. We let $\cg b$ represent the \emph{in-neighbourhood} of $b$, that is, $\cg b = \{a \in A \mid a \cg b\}$. 
Also, $\cg B := \Union_{b \in B} \cg b$. 
We let $\cg B \equiv \cg' B$ denote ``for all $a \in B$, $\cg a = \cg' a$'' (this notation is used frequently in subsequent proofs). Note that $\cg B \equiv \cg' B$ implies $\cg B = \cg' B$ but not the other way round.\footnote{Let $\cg a =\{a\}$ and $\cg b = \{a,b\}$, whereas $\cg' a = \cg' b = \{a,b\}$. Then $\cg \{a,b\} = \cg' \{a,b\}$, namely the set $\{a,b\}$, but $\cg \{a,b\} \not\equiv \cg' \{a,b\}$, because $\cg a \neq \cg' a$.} The \emph{identity relation} $I$ on the set of agents $A$ is $\{(a,a) \mid a \in A\}$. The \emph{universal relation} $U$ is $A \times A$.

A communication graph is a reflexive relation, because we assume that an agent always receives her own message. (An artifact of the relation is that we even assume reflexivity if the agent did not send a message.) But not every other agent may receive the message. Different communication graphs $\cg,\cg'$ in a communication pattern may inform a given agent $a$ in the same way ($\cg a = \cg' a$) even when agent $a$ is uncertain to what extent other agents inform each other (for $b \neq a$, $\cg b$ may differ from $\cg' b$). 

\begin{definition}[Updated epistemic model]
Given an epistemic model $M=(W,\sim,L)$ and a communication pattern $\cp$, the \emph{updated} epistemic model $M \odot \cp = (\dot W, \dot\sim, \dot L)$ of $M$ with $\cp$ is defined as: 
\[\begin{array}{lcl}
\dot W & = & W \times \cp \\
(w,\cg) \dot\sim_a (w',\cg') & \text{iff} &  w \sim_{\cg a} w' \text{ and } \cg a = \cg' a \\
\dot L(w,\cg) & = & L(w)
\end{array}\]
\end{definition}
The updated epistemic model encodes how the knowledge has changed after agents have informed each other according to communication pattern $\cp$. The new relation $\dot\sim_a$ for agent $a$ is the intersection $\sim_{\cg a}$ of the relations of all agents from which agent $a$ received messages.
 
In order to compare the information content of structures we now define bisimulation \cite{blackburnetal:2001}. For logics with distributed knowledge this notion is called \emph{collective bisimulation} \cite{Roelofsen07}.

\begin{definition}[Collective bisimulation]
A relation $Z$ between the domains of epistemic models $M = (W,\sim,L)$ and $M'=(W',\sim',L')$ is a {\em (collective) bisimulation}, notation $Z: M \bisim M'$, if for all $(w,w') \in Z$:
\begin{itemize}
\item {\bf atoms}: for all $p_a \in P$, $p_a \in L(w)$ iff $p_a \in L'(w')$;
\item {\bf forth}: for all nonempty $B \subseteq A$ and for all $v \in W$, if $w \sim_B v$ then there is $v'\in W'$ such that $(v,v') \in Z$ and $w' \sim_B v'$;
\item {\bf back}: for all nonempty $B \subseteq A$ and for all $v' \in W'$, if $w' \sim_B v'$ then there is $v\in W$ such that $(v,v') \in Z$ and $w \sim_B v$.
\end{itemize}
If there is a bisimulation $Z$ between $M$ and $M'$ we write $M \bisim M'$, and if there is one containing $(w,w')$ we write $(M,w)\bisim (M',w')$. We then say that $M$ and $M'$, respectively $(M,w)$ and $(M',w')$, are \emph{bisimilar}. 
\end{definition}
The (more standard) notion where {\bf forth} and {\bf back} only hold for singletons $B$ (that is, for individual agents) we here call {\em standard bisimulation} \cite{blackburnetal:2001}.

\begin{example}[Bisimilar but not collectively bisimilar] \label{ex.biscolbis}
We adapt the standard two-agent example, wherein the agents are uncertain about a propositional variable not known by either (which is not local), to a three-agent example wherein all agents know their local variables (which is local, as we require). Fig.~\ref{fig.bisnoncbis} shows two models for three agents $a,b,c$ that are (standardly) bisimilar but not collectively bisimilar. Note that agent $c$ has the identity relation on both models (and reflexive arrows are not drawn). 
%
%
\begin{figure}[h] 
\center
\begin{tikzpicture}
\node (m) at (-1.5,0) {$M:$};
\end{tikzpicture}
 \ 
\begin{tikzpicture}
\node (0) at (0,0) {$p_ap_b\overline{p}_c$};
\node (0w) at (3,0.5) {$w$};
\node (1) at (3,0) {$p_ap_bp_c$};
\draw[-] (0) -- node[above] {$a,b$} (1);
\end{tikzpicture}
\quad
\begin{tikzpicture}
\node (ma) at (5.5,0) {$M':$};
\node (ma) at (5.5,-1) {\color{white}$M':$};
\end{tikzpicture}
 \ 
\begin{tikzpicture}
\node (00) at (7,0) {$p_ap_b\overline{p}_c$};
\node (10) at (10,0) {$p_ap_bp_c$};
\node (01) at (7,2) {$p_ap_bp_c$};
\node (11) at (10,2) {$p_ap_b\overline{p}_c$};
\node (00w) at (9.5,0.5) {$w'$};
\draw[-] (00) -- node[above] {$a$} (10);
\draw[-] (01) -- node[above] {$a$} (11);
\draw[-] (00) -- node[left] {$b$} (01);
\draw[-] (10) -- node[right] {$b$} (11);
\end{tikzpicture}
\caption{Bisimilar but not collectively bisimilar}
\label{fig.bisnoncbis}
\end{figure}
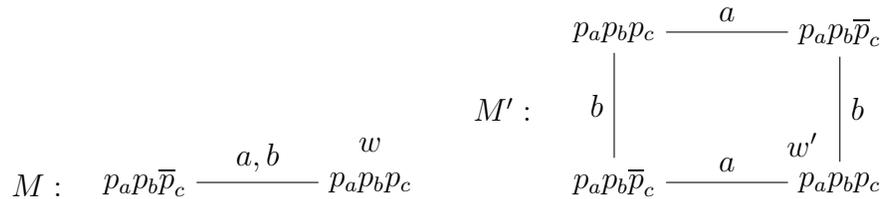
\end{example}

Let us see some examples of communication patters, and how to update epistemic models with them. 

\begin{example}[Byzantine attack] \label{example.byzantine}
This example is taken from \cite{diego:2021} and models Byzantine attack \cite{lamportetal:1982,DworkM90}. We let $A = \{a,b\}$. Generals $a$ and $b$ wish to schedule an attack, where $b$ desires to learn whether $a$ wants to `attack at dawn' ($p_a$) or `attack at noon' ($\neg p_a$). General $a$ now sends her decision to general $b$ in a message that may fail to arrive. This fits the communication pattern $\cp = \{I,R^{ab}\}$ where $R^{ab} = I \union \{(a,b)\}$, which models that $a$ is uncertain whether her message has been received by $b$. In this instantiation of Byzantine generals, general $b$ has no local variable. See Fig.~\ref{fig.a}. We note that the `communication pattern model' $A'$ of \cite[Fig.~1]{diego:2021} has the same update effect.

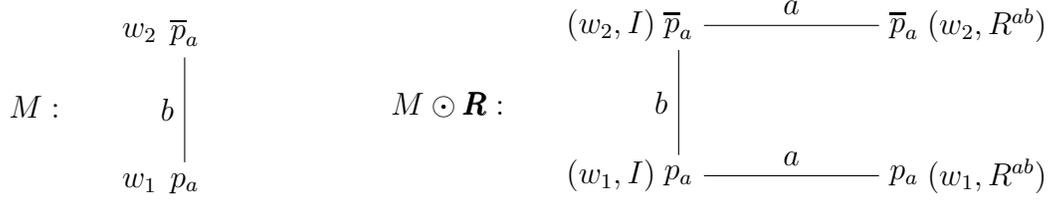
\begin{figure}[h]
\center
\begin{tikzpicture}
\node (m) at (-5,0) {$M:$};
\node (m) at (-5,-1) {\color{white}$M:$};
\end{tikzpicture} %
\quad 
\begin{tikzpicture}
\node (bm00) at (-4.6,0) {$w_1$};
\node (am01) at (-4.6,2) {$w_2$};
\node (m00) at (-4,0) {$p_a$};
\node (m01) at (-4,2) {$\np_a$};
\draw[-] (m00) -- node[left] {$b$} (m01);
\end{tikzpicture} 
%
\hspace{2cm}
\begin{tikzpicture}
\node (m) at (0,0) {$M \odot \cp:$};
\node (m) at (0,-1) {\color{white}$M \odot \cp:$};
\end{tikzpicture} %
\quad 
\begin{tikzpicture}
\node (b00) at (1.1,0) {$(w_1, I)$};
\node (b10) at (6.1,0) {$(w_1, R^{ab})$};
\node (a01) at (1.1,2) {$(w_2, I)$};
\node (a11) at (6.1,2) {$(w_2,R^{ab})$};
\node (00) at (2,0) {$p_a$};
\node (10) at (5,0) {$p_a$};
\node (01) at (2,2) {$\np_a$};
\node (11) at (5,2) {$\np_a$};
\draw[-] (00) -- node[above] {$a$} (10);
\draw[-] (00) -- node[left] {$b$} (01);
\draw[-] (01) -- node[above] {$a$} (11);
\end{tikzpicture}
\caption{Byzantine attack}\label{fig.a}
\end{figure}
The visual conventions in Fig.~\ref{fig.a} are as in Fig.~\ref{annebill}: states are labelled with valuations, where $\overline{p}_a$ means that $p_a$ is false. States indistinguishable for an agent are connected by a link labelled with that agent. We now also have named the states. Let us verify some updated links. In the updated model $M \odot \cp$, for example, $(w_1,I) \sim_b (w_2,I)$ because $w_1 \sim_b w_2$ and $I b = I b = \{b\}$, whereas $(w_1,R^{ab}) \not\sim_b (w_2,R^{ab})$ because, although $R^{ab} b = \{a,b\}$, we  have that $w_1 \not\sim_{ab} w_2$. Relation $\sim_{ab}$ is the identity relation on both models.
\end{example}

\begin{example}[Immediate Snapshot] \label{snapshot}
This communication pattern can be seen as the set of ordered partitions of $A$. An ordered partition is a list of mutually exclusive subsets of $A$ such that their union is $A$. In distributed computing such an ordered partition is called a \emph{schedule}, where each element in the partition is called a \emph{concurrency class}. It is a standard way to model asynchronous communication in rounds by way of a shared memory \cite{herlihyetal:2013}. For example, if $A = \{a,b\}$ then the set of ordered partitions is represented by $\{R^{ab},R^{ba},U\}$, where $R^{ba} = I \union \{(a,b)\}$ and $R^{ba} = I \union \{(b,a)\}$. The list notation for these three partitions is respectively $a.b$, $b.a$ and $ab$. The order in this partition is the order in which the agents write on and read from the shared memory. For example, in $a.b$, first $a$ writes her value and reads all written values, which at this stage is only her own value. After that $b$ writes his value and reads all written values, which now are the values of $a$ and $b$. In other words, $b$ receives $a$'s message but not vice versa: $R^{ab}$. For $b$ this is indistinguishable from the (trivial) partition $ab$ wherein $a$ and $b$ simultaneously write and read, and wherein $b$ also reads the values of $a$ and $b$: $U$.
\end{example}

\begin{example}[Fully asynchronous] \label{asynch}
This is the communication pattern $\cp = \{ R \in \power(A\times A) \mid I \subseteq R \}$.  It consists of all communication graphs. It represents complete absence of information on the arrival of messages with other agents.
\end{example}

\begin{example}[No message arrives] \label{none}
This is the communication pattern $\cp = \{I\}$. We now have that for any model $M$, $M \odot \cp$ is bisimilar (and even isomorphic) to $M$, via the relation $Z$  linking each $w \in W$ to $(w,I)$. After a failed round of message transmission the knowledge of the agents has not changed. 
\end{example}
%
It is important to note the difference between the communication pattern  $\{I\}$ of Ex.~\ref{none} where there is certainty about transmission failure embodied in $I$, from communication patterns \emph{containing} $I$, such as the Byzantine attack of Ex.~\ref{example.byzantine} or the asynchrony of Ex.~\ref{asynch}. Pattern $\{I\}$ means that it is common knowledge among the agents that no message was received. It therefore intuitively means that no message was \emph{sent}. Whereas the pointed communication pattern $(\{\cg^{ab},I\},I)$ of Byzantine attack intuitively means that agent $a$ sent the message but that it was not received by $b$, and that $b$ did not send a message. In pointed communication pattern $(\{ R \in \power(A\times A) \mid I \subseteq R \},I)$, all agents sent a message and no message was received by anyone. Now that is transmission failure!

\begin{example}[Public Announcement] \label{ex.pa}
A \emph{public announcement} wherein all agents announce all they know is the communication pattern $\cp = \{U\}$ (and similarly, an announcement by the agents in $B \subseteq A$ of all they know, would be the communication pattern $\{(B\times A) \union I\}$; that is also known as the (public) \emph{group  announcement} \cite{agotnesetal.jal:2010} by the agents in $B$). It is not an ordinary public announcement as in the muddy children problem \cite{mosesetal:1986}. Let there be children $a,b,c$ and atoms $m_a,m_b,m_c$ for `$a/b/c$ is muddy', and let $a$ be clean and $b$ and $c$  muddy. The father's initial announcement of $m_a \vel m_b \vel m_c$ that at least one child is muddy is already problematic, because it is not made by an agent modelled in the system. But the children subsequently announcing that they do not know whether they are muddy is, although an announcement, not all they know. That would be the announcement $K_a (m_b \et m_c) \et K_b (\neg m_a \et m_c) \et K_c (\neg m_a \et m_b)$ (where $K_a$ means `agent $a$ knows' and $\et$ is conjunction). This  removes all uncertainty in the model and restricts it to a singleton wherein it is common knowledge that $a$ is clean and $b,c$ are muddy.

This formalization of public announcement corresponds to what is known as the \emph{refinement semantics} for public announcement \cite{jfaketal.jancl:2007}, not to the (usual) domain restriction semantics \cite{plaza:1989} or the relational restriction semantics \cite{gerbrandyetal:1997}.
\end{example}

\begin{example}[Gossip]
In a \emph{gossip protocol} \cite{kermarrecetal:2007,hedetniemietal:1988} the agents communicate by peer-to-peer calls between neighbours given a network where nodes stand for agents and links between nodes connect neighbours. We propose to model this as the communication pattern $\cp = \{ab \mid a,b \in A, a \neq b\}$, where $ab:= \{(a,b),(b,a)\} \union I$. Communication graph $ab$ represents the call between $a$ and $b$.  For example, if $A = \{a,b,c,d\}$, the relation $ab$ is $\{(a,a),(a,b),(b,a),(b,b),(c,c),$ $(d,d)\}$. Given our assumption that communication happens in rounds, this is {\em synchronous gossip}.


Various modes of communication are common in gossip. The one above wherein the callers exchange all their secrets (and all they know) is called \emph{pushpull}. Other common modes are \emph{push} wherein only the caller informs the person called, and \emph{pull} wherein only the person called informs the caller. \emph{Push} is the pattern wherein call $ab:= \{(a,b)\} \union I$ and \emph{pull} is the pattern where call $ab:= \{(b,a)\} \union I$. {\em Asynchronous gossip} (for any of `pushpull', `push', and `pull') results when adding `no call' to the pattern above, that is, we then get $\cp' = \{ab \mid a,b \in A, a \neq b\} \union \{I\}$. This makes the agents uncertain if a call took place.
\end{example}

We can now better compare our work with other research. 

\emph{Resolving distributed knowledge} \cite{AgotnesW17} is a singleton communication pattern for graph $(B \times B) \union I$, where $B \subseteq A$. Note the difference with the agents in $B$ resolving what they distributedly know (not received by those not in $A$) and publicly announcing it, the graph $(B \times A) \union I$ (Ex.~\ref{ex.pa}).  When $A=B$ the two coincide.

A communication graph is a \emph{reading event} in \cite{Baltag20}. Whereas a  communication pattern is an \emph{arbitrary reading event} \cite[Sect.~6]{Baltag20}. An arbitrary reading event is a Kripke frame of which the domain elements are decorated with communication graphs and where $\cg \sim_a \cg'$ iff $\cg a = \cg' a$. There are arbitrary reading events that are not communication patterns (namely when different domain elements are decorated with the same communication graph).

In \cite{diego:2021}, that introduced communication patterns, they are represented as Kripke frames, just as in \cite{Baltag20}. They may also have, like action models, executability \emph{preconditions}, in order to formalize dynamics in non-oblivious distributed systems. 

\section{Language, semantics, and axiomatization} \label{sec.logic}

\subsection{Language and semantics}

We first define the logical language and semantics.
\begin{definition}[Language] Given $A$ and $P$, the language $\lang$ is defined by a BNF (where $p_a \in P$, $B \subseteq A$, $\cp \subseteq \powerset(A \times A)$, and $\cg \in \cp$).
\[ \phi := p_a \mid \neg \phi \mid \phi \et \phi \mid D_B \phi \mid [\cp,\cg]\phi \]
\end{definition}
Expression $D_B \phi$ is read as `the agents in $B$ have distributed knowledge of $\phi$'. We write $K_a \phi$ for $D_{\{a\}}\phi$, for `agent $a$ knows $\phi$'.  We let $\lang^-$ stand for the language without the construct  $[\cp,\cg]\phi$.

\begin{definition}[Semantics on epistemic models] Given $M = (W,\sim,L)$ and $w \in W$,  we define the \emph{satisfaction relation} $\models$ by induction on $\phi\in\lang$ (where $p \in P$, $a \in A$, $B \subseteq A$, $\cp$ a communication pattern, and $\cg\in\cp$).
\[ \begin{array}{lcl}
M,w \models p_a & \text{iff} & p_a \in L(w)\\
M,w \models \neg\phi & \text{iff} & M,w \not\models \phi\\
M,w \models \phi\et\psi & \text{iff} & M,w \models \phi \text{ and } M,w \models \psi \\
M,w \models D_B \phi & \text{iff} & M,v \models \phi \text{ for all } v \sim_B w \\
M,w \models [\cp,\cg]\phi & \text{iff} & M \odot \cp, (w,\cg) \models \phi \\
\end{array} \]
Formula $\phi$ is {\em valid} on $M$ iff for all $w \in W$, $M,w \models \phi$; and formula $\phi$ is {\em valid} iff for all pointed models $(M,w)$, $M,w \models \phi$. 
\end{definition}

The (required) locality of epistemic models causes distributed knowledge to have slightly different properties in our semantics. In the standard semantics of distributed knowledge $D_B \phi \eq \phi$ is invalid for any $B \subseteq A$. Whereas in our semantics $D_A \phi \eq \phi$ is valid although $D_B \phi \eq \phi$ for $B \subset A$ remains invalid. It is easy to show that $D_A \phi \eq \phi$ is valid by induction on the structure of $\phi$. The base case holds because for all $M$ and $w,v \in \domain(M)$, and for all local variables $p_a$: $M,w \models p_a $ iff $M,v \models p_a$, whenever $w \sim_A v$  (even when $\sim_A$ is not the identity relation). Whereas we need not have that, if  $w \sim_B v$ and $a \notin B$.

Informally, the information content of a pointed epistemic model is the set of all formulas that are true in that model. Formally, given models $M = (W,\sim,L)$ and $M'= (W',\sim',L')$, and $w \in W$, $w'\in W'$, we denote ``for all $\phi \in \lang$, $M,w \models\phi$ iff $M',w'\models\phi$'' as $(M,w) \equiv (M',w')$; we then say that $(M,w)$ and $(M',w')$ are {\em modally equivalent}. Analogously we define $(M,w) \equiv^- (M',w')$ for the language $\lang^-$. It is known that bisimilarity implies modal equivalence with respect to $\lang^-$ \cite{Roelofsen07}. This also holds for the language extended with communication patterns:

\begin{theorem} \label{theorem.xx}
$(M,w)\bisim (M',w')$ implies $(M,w)\equiv (M',w')$
\end{theorem}

\begin{proof}
We prove by induction on $\phi$ that ``for all $\phi$, for all pointed models $(M,w)$, $(M',w')$: $M,w\models\phi$ iff $M',w'\models\phi$''. All cases are standard for logics of distributed knowledge (namely as in \cite{Roelofsen07}) except of course for the update case $[\cp,\cg]\phi$, to which we therefore restrict the proof.

We show that ``for all pointed models $(M,w)$, $(M',w')$: if $(M,w)\bisim(M',w')$ then $M,w \models [\cp,\cg]\phi$ iff $M',w' \models [\cp,\cg]\phi$''. Let $M=(W,\sim,L)$, $M'=(W',{\sim'}, {L'})$, (and further below) $M\odot\cp=(W\times\cp,\dot\sim,\dot L)$, and $M'\odot\cp=(W'\times\cp,\dot \sim', \dot L')$. We show the direction from left to right. The other direction is shown similarly.

Assume $M,w \models [\cp,\cg]\phi$. Then $M \odot \cp, (w,\cg) \models \phi$. Given $(M,w)\bisim (M',w')$, let $Z$ be a bisimulation between $M$ and $M'$ containing $(w,w')$. We claim that there is a bisimulation between $(M\odot\cp, (w,\cg))$ and $(M'\odot\cp,(w',\cg))$, namely $\dot Z$ defined as: $((w,\cg),(w',\cg')) \in \dot Z$ iff $(w,w') \in Z$ and $\cg=\cg'$.

Let now $((w,\cg),(w',\cg')) \in \dot Z$ be arbitrary. We check the requirements of bisimulation.

{\bf atoms}: This is obvious as $\dot L(w,\cg)=L(w)= L'(w')= \dot L'(w',\cg)$.

{\bf forth}:  Assume $(w,\cg) \dot\sim_B (v,S)$. Then for all $a \in B$, $(w,\cg) \dot\sim_a (v,S)$, so, by definition of $\dot\sim_a$, for all $a \in B$, $w \sim_a v$ and $\cg a = S a$. A different way to write the latter is $w \sim_B v$ and $\cg B \equiv S B$. From $w \sim_B v$ and $(w,w') \in Z$ and {\bf forth} for $B$ it follows that there is a $v' \in W'$ such that $(v,v') \in Z$ and $w' \sim'_B v'$. From $w' \sim'_B v'$ and $\cg B \equiv S B$ it then follows that $(w',\cg) \dot\sim'_B (v',S)$. From $(v,v') \in Z$ and the definition of $\dot Z$ it follows that $((v,\cg),(v',S)) \in \dot Z$.

{\bf back}: Similar to {\bf forth}.

Having established that $(M\odot\cp, (w,\cg))$ is bisimilar to $(M'\odot\cp,(w',\cg))$, we now apply the induction hypothesis for $\phi$ on pointed models $(M\odot\cp, (w,\cg))$ and $(M'\odot\cp,(w',\cg))$ thus obtaining that $M' \odot \cp, (w',\cg) \models \phi$, so that $M',w' \models [\cp,\cg]\phi$ as required.
\end{proof}
On image-finite models we also have that $(M,w)\equiv (M',w')$ implies $(M,w)\bisim$ $(M',w')$, and therefore the Hennessy-Milner property \cite{blackburnetal:2001}. (An epistemic model $M = (W, \sim,L)$ is image-finite if for all $w \in W$ and $a \in A$, $[w]_a$ is finite.) It is known that $(M,w)\equiv^- (M',w')$ (in the restricted language $\lang^-$) implies $(M,w)\bisim (M',w')$ \cite{Roelofsen07}. As $(M,w)\equiv (M',w')$ implies $(M,w)\equiv^- (M',w')$, therefore also $(M,w)\equiv (M',w')$ implies $(M,w)\bisim (M',w')$.

\begin{example}
The models $(M,w)$ and $(M',w')$ from Ex.~\ref{ex.biscolbis} (Fig.~\ref{fig.bisnoncbis}) satisfy different formulas in the logic of distributed knowledge. For example, $M,w \not\models D_{ab} p_c$ whereas $M',w' \models D_{ab} p_c$. With Th.~\ref{theorem.xx} it therefore follows (by contraposition) that they are not collectively bisimilar. This was indeed observed in Ex.~\ref{ex.biscolbis}.
\end{example}

\subsection{Axiomatization} \label{axiomatization}

We proceed with the axiomatization. The axiomatization of the logic of communication patterns is that of the logic of distributed knowledge expanded with reduction axioms $\mathbf {C}^1$---$\mathbf{C}^4$ and rule $\mathbf{N}^\odot$ involving communication patterns, the axiom {\bf L} describing that agents know their local state, and the auxiliary derivation rule {\bf RE} (replacement of equivalents). The axiomatization is displayed in Table~\ref{table}. A \emph{derivation} is a sequence of formulas such that every formula is the instantiation of an axiom, or the conclusion of an instantiation of a derivation rule where the premisses (or premiss) are prior formulas in the sequence. A formula occurring in a derivation is a \emph{theorem}.

\begin{table}[h]
\[\begin{array}{ll}
\mathbf{P} & \text{all instances of prop.\ tautologies} \\
\mathbf{L} & K_a p_a \vel K_a \neg p_a \\
\mathbf{K}^D & D_B (\phi \imp \psi) \imp D_B \phi \imp D_B \psi \\
\mathbf{T}^D & D_B \phi \imp \phi \\
\mathbf{4}^D & D_B \phi \imp D_B D_B \phi \\
\mathbf{5}^D & \neg D_B \phi \imp D_B \neg D_B \phi \\
\mathbf{W} & D_B \phi \imp D_C \phi \\
\mathbf{C}^1 & {[\cp,\cg]} p_a \eq p_a \\
\mathbf{C}^2 & {[\cp,\cg]} \neg\phi \eq  \neg [\cp,\cg] \phi \\
\mathbf{C}^3 & {[\cp,\cg]} (\phi\et\psi) \eq ([\cp,\cg] \phi \et [\cp,\cg] \psi) \\
\mathbf{C}^4 & {[\cp,\cg]} D_B \phi \eq \Et_{\cg' B \equiv \cg B} D_{\cg B} [\cp,\cg'] \phi \\ \ \\
\mathbf{MP} &  \text{From } \phi \imp \psi \text{ and } \phi \text{ infer } \psi \\
\mathbf{N}^D &  \text{From } \phi \text{ infer } D_B \phi \\
\mathbf{N}^\odot & \text{From } \phi \text{ infer } [\cp,\cg] \phi \\ 
\mathbf{RE} & \text{From } \phi \eq \psi \text{ infer } \chi[p_a/\phi] \eq \chi[p_a/\psi] \\
\end{array}\]
\caption{Axiomatization of communication pattern logic, where $B, C \subseteq A$ with $B \subseteq C$, $a \in A$, and $p_a \in P_a$. We recall that $K_a$ abbreviates $D_{\{a\}}$.}
\label{table}
\end{table}

In the derivation rule {\bf RE}, $\chi[p_a/\phi]$ stands for uniform substitution of the occurrences of atom $p_a$ in formula $\chi$ by $\phi$. The reduction axioms  for communication patterns resemble those for action models \cite{baltagetal:1998} with trivial preconditions, except for the reduction axiom for distributed knowledge after update. There is no reduction axiom for a sequence of two communication patterns. That explains the presence of the derivation rule {\bf RE}, which is not needed in the axiomatization of distributed knowledge, as it is then admissible.

The validity of all axioms and the validity preservation of all rules is obvious, except for  that of $\mathbf{C}^4$, and maybe that of $\mathbf{L}$, $\mathbf{N}^\odot$ and $\mathbf{RE}$. There are therefore shown.

\begin{proposition}[$\mathbf{L}$] \label{l}
$\models K_a p_a \vel K_a \neg p_a$.
\end{proposition}
\begin{proof}
Let epistemic model $M = (W,\sim,L)$ and state $w \in W$ be arbitrary. Let now $v \in W$ be such that $v \sim_a w$. As epistemic models are local, $L(v)=L(w)$, in other words, $p_a \in L(v)$ iff $p_a \in L(w)$. Therefore, if $M,w \models p_a$, then $M,v \models p_a$, so that $M,w \models K_a p_a$ as $v$ was arbitrary. Otherwise, if $M,w \models \neg p_a$, then  $M,v \models \neg p_a$, so that $M,w \models K_a \neg p_a$. Either way we obtain $M, w \models K_a p_a \vel K_a \neg p_a$.
\end{proof}

\begin{proposition}[$\mathbf{C}^4$] \label{prop.ddd}
$\models [\cp,\cg] D_B \phi \eq \Et_{\cg' B \equiv \cg B} D_{\cg B} [\cp,\cg'] \phi$. 
\end{proposition}

\begin{proof}
Let $M = (W,\sim,L)$ and $w \in W$. Then: 

\bigskip

\noindent $
M,w \models [\cp,\cg] D_B \phi \quad \\
\Eq \text{(semantics of communication patterns)} \\
M \odot \cp, (w,\cg) \models D_B \phi \quad \\
\Eq \text{(semantics of distributed knowledge)}  \\
\forall w' \in W, \forall \cg' \in \cp:  (w,\cg) \sim_B (w',\cg') \Imp M \odot \cp, (w',\cg') \models \phi \quad \\
\Eq \\
\forall w' \in W, \forall \cg' \in \cp: (\forall a \in B: (w,\cg) \dot\sim_a (w',\cg')) \Imp M \odot \cp ,(w',\cg') \models \phi \quad \\
\Eq \\
\forall w' \in W, \forall \cg' \in \cp: (\forall a \in B\!: \! w \sim_{\cg a} w' \ \& \ \cg a = \cg' a) \!\Imp\! M \odot \cp ,(w',\cg') \models \phi \  \\
\Eq \\
\forall w' \in W, \forall \cg' \in \cp: (\forall a \in B: w \sim_{\cg a} w' \ \& \ \forall a \in B: \cg a = \cg' a) \Imp M \odot \cp ,(w',\cg') \models \phi \quad \\
\Eq \text{($\cg B \equiv \cg' B$ is defined as ``for all $a \in B$, $\cg a = \cg' a$'')} \\
\forall w' \in W, \forall \cg' \in \cp: (w \sim_{\cg B} w' \ \& \ \cg B \equiv \cg' B) \Imp M \odot \cp ,(w',\cg') \models \phi \quad \\
\Eq \\
\forall \cg'\in\cp : \cg B \equiv \cg' B \Imp (\forall w'\in W: w \sim_{\cg B} w' \Imp M \odot \cp, (w,\cg') \models \phi) \quad \\
\Eq \text{(semantics of communication patterns)} \\
\forall \cg'\in\cp : \cg B \equiv \cg' B \Imp (\forall w'\in W: w \sim_{\cg B} w' \Imp M,w \models [\cp,\cg'] \phi) \quad \\
\Eq  \text{(semantics of distributed knowledge)} \\
\forall \cg'\in\cp : \cg B \equiv \cg' B \Imp M,w \models D_{\cg B} [\cp,\cg'] \phi \quad \\
\Eq \\
M,w \models \Et_{\cg' B \equiv \cg B} D_{\cg B} [\cp,\cg'] \phi 
$
\end{proof}

The interaction between communication patterns and distributed knowledge of Prop.~\ref{prop.ddd} is reminiscent of \cite[Prop.\ 5]{WangA13}, it is similar to \cite[Prop.\ 4.6]{Baltag20}, and also somewhat similar to \cite[Prop.\ 2, items 6 \& 7]{AgotnesW17}. 

In \cite[Prop.~5]{WangA13}, axiom $[\phi]D_B\psi \eq (\phi \imp D_B [\phi]\psi)$ describes the interaction between public announcement $[\phi]$ and distributed knowledge. A public announcement wherein a group $B \subseteq A$ of agents simultaneously reveal all they know corresponds to the communication pattern we also called public announcement (Ex.~\ref{ex.pa}). 

In \cite{Baltag20}, axiom $[!\alpha] D_B \phi \eq D_{\alpha(B)}[!\alpha]\phi$ (where $!\alpha$ is a {\em reading map}) is the special case of $\mathbf{C}^4$ for a singleton communication pattern, which would be: $[\cp,\cg] D_B \phi \eq D_{\cg B} [\cp,\cg] \phi$. This is not surprising, as their axiom inspired our slightly more general axiom.

In \cite{AgotnesW17} we find the validities $R_B D_C \phi \eq D_{B \union C}R_B\phi$ for $B \inter C \neq \emptyset$, and $R_B D_C \phi \eq D_C R_B\phi$ for $B \inter C = \emptyset$. They call $R_B$ the {\em resolution operator}. It can also be seen as a special case of the reading map $!\alpha$ in \cite{Baltag20}.

\begin{proposition}[$\mathbf{N}^\odot$] \label{nodot} 
If $\models \phi$, then $\models [\cp,\cg] \phi$. 
\end{proposition}
\begin{proof}
In order to show that $\models [\cp,\cg] \phi$, let $M = (W,\sim,L)$ and $w \in W$ be given. We wish to show that $M,w \models [\cp,\cg] \phi$. By definition of the semantics, this is equivalent to $M \odot \cp, (w,\cg) \models \phi$. As we assumed that $\phi$ is valid, this holds.
\end{proof}

\begin{proposition}[$\mathbf{RE}$] \label{re}
If $\models \phi \eq \psi$, then $\models \chi[p_a/\phi] \eq \chi[p_a/\psi]$.
\end{proposition}
\begin{proof}
This is proved by induction on $\chi$. The cases of interest are $\chi = (D_B \chi')[p_a/\phi]$ and $\chi = ([\cp,\cg]\chi')[p_a/\phi]$. 

\bigskip

\noindent $
M,w \models (D_B \chi')[p_a/\phi] \\
\Eq \text{(apply uniform substitution)} \\
M,w \models D_B \chi'[p_a/\phi] \\
\Eq \text{(semantics of distributive knowledge)}  \\
M,v \models \chi'[p_a/\phi] \text{ for all } v \sim_B w \\
\Eq \text{(inductive hypothesis)} \\
M,v \models \chi'[p_a/\psi] \text{ for all } v \sim_B w \\
\Eq \text{(similar to the first two steps above, but in the other direction)}  \\
M,w \models (D_B \chi')[p_a/\psi]$

\bigskip
\medskip

\noindent $
M,w \models ([\cp,\cg]\chi')[p_a/\phi] \\
\Eq \text{(apply uniform substitution)} \\
M,w \models [\cp,\cg]\chi'[p_a/\phi] \\
\Eq \text{(semantics of communication patterns)}  \\
M \odot \cp, (w,\cg) \models \chi'[p_a/\phi] \\
\Eq \text{(inductive hypothesis)} \\
M \odot \cp, (w,\cg) \models \chi'[p_a/\psi] \\
\Eq \text{(similar to the first two steps above, but in the other direction)}  \\
M,w \models ([\cp,\cg]\chi')[p_a/\psi]$
\end{proof}

The logic of communication patterns does not satisfy the so-called {\em  substitution property} (from $\phi$ infer $\phi[p/\psi]$ for any $\psi$). This is because some axioms feature propositional variables instead of arbitrary formulas. For a simple example, $K_a p_a \vel K_a \neg p_a$ is valid, but $K_a p_b \vel K_a \neg p_b$ for $b \neq a$ is invalid.

\begin{theorem} \label{prop.bal}
The axiomatization of communication pattern logic in Table~\ref{table} is sound and complete.
\end{theorem}

\begin{proof}
The soundness follows from the literature on distributed knowledge (see, e.g., \cite{handbookintro:2015}) and from the above Propositions~\ref{l}, \ref{prop.ddd}, \ref{nodot}, and \ref{re}.

The completeness follows from $(i)$ the completeness of the logic of distributed knowledge  \cite{Roelofsen07,WangA13,handbookintro:2015}, $(ii)$ the admissibility of the derivation rule {\bf RE} in that axiomatization (as in \cite{AgotnesW17}, for a similar setting), $(iii)$ termination of an inside-out reduction showing that formulas containing communication pattern modalities are provably equivalent to formulas without communication patterns (where {\bf RE} is essential, as explained in general in \cite{WC13}), and $(iv)$ lack of interference with locality axiom {\bf L}. 
\end{proof}
Let us sketch the role of the terminating reduction in determining whether a given formula in $\phi\in\lang$ is a theorem:

If $\phi$ does not contain communication pattern modalities, determine whether it is a theorem in the logic of distributed knowledge. Otherwise, take an innermost subformula of $\phi$ of shape $[\cp,\cg]\psi$ (that is, a subformula  $[\cp,\cg]\psi$ of $\phi$ such that $\psi\in \lang^-$). Determine a formula $\psi'\in\lang^-$ such that $\psi' \eq [\cp,\cg]\psi$ is a theorem. This is done by way of applying the reduction axioms and necessitation for communication patterns, and epistemic and propositional axioms and rules.\footnote{The axioms $\mathbf{C}^2$---$\mathbf{C}^4$ all have a shape where on the left-hand side the communication pattern binds a formula of higher complexity than the formula bound by the communication pattern on the right-hand side, whereas in $\mathbf{C}^1$ the communication pattern has disappeared on the right-hand side. By successively applying these axioms we can `make a communication pattern modality disappear'. An example is: $[\cp,\cg](p_a \et \neg p_b)$, iff $[\cp,\cg]p_a \et [\cp,\cg]\neg p_b$, iff $p_a \et [\cp,\cg]\neg p_b$, iff $p_a \et \neg [\cp,\cg] p_b$, iff $p_a \et \neg p_b$.} Let now $\phi'$ be $\phi$ wherein $[\cp,\cg]\psi$ is substituted by $\psi'$. This is an application of {\bf RE}. Formula $\phi'$ contains one fewer communication pattern modality. We now use induction on the number of such modalities in the formula, and thus obtain a $\phi'' \in \lang^-$ provably equivalent to $\phi \in \lang$. Then, determine where $\phi''$ is a theorem in the logic of distributed knowledge.

More formal details could be given. However, the completeness result is a generalization of the completeness of the similar logic of resolving distributed knowledge in \cite{AgotnesW17}, and a generalization of the completeness of the similar logic with reading map modalities in \cite{Baltag20}, and also a special case of the conjectured  completeness of the logic with arbitrary reading events in \cite{Baltag20}. Rather than presenting our complete axiomatization as an original result, with exhaustive proofs, we credit it to the authors of \cite{AgotnesW17} and of \cite{Baltag20}, and refer to their proof details.

Locality axioms expressing that agents know their local variables do not occur in \cite{AgotnesW17,Baltag20} but are common for distributed systems. For example, an interpreted system \cite{faginetal:1995}, wherein indistinguishability is determined by the agent's local state, satisfies locality. 

We provided an inside-out reduction instead of an outside-in reduction because the composition of two communication patterns may not be a communication pattern. Such questions concerning update expressivity are succinctly discussed in the final Sect.~\ref{sec.comparison}, also in relation to action models \cite{baltagetal:1998}, that are closed under composition. A similar observation is made concerning resolving distributed knowledge: `there does not seem to be a reduction axiom in this case' \cite[Sect.~3.1.2]{AgotnesW17}, that is, for sequentially resolving distributed knowledge for two different subsets $B,C \subseteq A$ of the set of all agents. 

As the complete axiomatization of communication pattern logic is a reduction system, from Theorem~\ref{prop.bal} already indirectly follows Theorem~\ref{theorem.xx} that collective bisimilarity implies modal equivalence {\em in the extended language}. We prove this by contraposition:

Assume $(M,w) \not\equiv(M',w')$. Then there is a $\phi\in\lang$ such that $M,w \models\phi$ and $M',w'\not\models\phi$. As  $\phi$ is equivalent to a $\phi'\in\lang^-$, also $M,w \models\phi'$ and $M',w'\not\models\phi'$, so that $(M,w) \not\equiv^-(M',w')$. Collective bisimilarity implies modal equivalence in $\lang^-$ \cite{Roelofsen07}. Therefore, using contraposition on that, $(M,w)\not\hspace{-.2cm}\bisim(M',w')$.

\section{Simplicial complexes and communication patterns} \label{sec.simp}

In this section we update simplicial models with communication patterns. First, we need to introduce standard terminology around simplicial complexes.

\subsection{Simplicial models and their updates}
Assume agents $A$ and local variables $P \subseteq P'\times A$ as before.

\begin{definition}[Simplicial complex]
Given a set of \emph{vertices} $V$, a \emph{(simplicial) complex} $C$ is a set of nonempty finite subsets of $V$, called \emph{simplices}, that is closed under subsets (such that for all $X \in C$, $Y \subseteq X$ implies $Y \in C$), and that contains all singleton subsets of $V$. 
\end{definition} 
If $Y \subseteq X$ we say that $Y$ is a \emph{face} of $X$. A maximal simplex in $C$ is a \emph{facet}. The facets of a complex $C$ are denoted as $\FF(C)$, and the vertices of a complex $C$ are denoted as $\VV(C)$. 
The \emph{dimension} of a simplex $X$ is $|X|-1$. The dimension of a complex is the maximal dimension of its facets. A simplicial complex is \emph{pure} if all facets have the same dimension. 


We decorate the vertices of simplicial complexes with agent's names, that we often refer to as  \emph{colours}. A \emph{chromatic map} $\chi: \VV(C) \imp A$ assigns colours to vertices such that different vertices of the same simplex are assigned different colours.  Thus, $\chi(v)=a$ denotes that the vertex $v$ belongs to agent $a$. For any simplex $X \in C$, $\chi(X) := \{ \chi(v) \mid v \in X \}$. 
A pair $(C,\chi)$ consisting of a simplicial complex $C$ and a chromatic map $\chi$ is a \emph{chromatic simplicial complex}. From now on, all simplicial complexes will be pure chromatic simplicial complexes. 
%
%
We also decorate the $a$-coloured vertices of simplicial complexes  with subsets of $P_a$.  \emph{Valuations} (valuation functions) assigning sets of local variables for agents $a$ to vertices coloured $a$ are denoted $\ell, \ell', \dots$ So, $\ell(v) \subseteq P_{\chi(v)}$.
 For any $X \in C$, $\ell(X)$ stands for $\Union_{v \in X} \ell(v)$. 
\begin{definition}[Simplicial model]
A \emph{simplicial model} $\C$ is a triple $(C,\chi,\ell)$ where $C$ is a simplicial complex, $\chi$ is a chromatic map, and $\ell$ is a valuation function. A {\em pointed simplicial model} is a pair $(\C,X)$ with $X \in C$.
\end{definition} 

\begin{definition}[Update of a simplicial model]
Given a simplicial model $\C=(C,\chi,\ell)$ and a communication pattern $\cp$, the update $\C \oslash \cp = (\ddot C, \ddot \chi, \ddot \ell)$ of $\C$ with $\cp$ is defined as follows.  Each facet $Y \in \FF(C)$ and communication graph $\cg \in \cp$ determines a facet denoted $Y_\cg$ of $\ddot C$ that is defined as $Y_\cg := \{ (v,X) \mid v \in Y, X \subseteq Y, \chi(X) = \cg \chi(v) \}$. The complex $\ddot C$ consists of the faces of these facets. Note that this also determines the vertices $\VV(\ddot C)$. For any $(v,X)\in\VV(\ddot C)$ we then define that $\ddot\chi(v,X)= \chi(v)$ and $\ddot\ell(v,X)=\ell(v)$.
\end{definition}
We will give examples of updates of simplicial models after introducing bisimulation and the semantics.

\subsection{Collective bisimulation between simplicial models}

We now define collective bisimulation between simplicial models, as a generalization of the bisimulation between simplicial models proposed in \cite{GoubaultLR21,ledent:2019,hvdetal.simpl:2022}. The latter has {\bf forth} and {\bf back} only for individual agents $a \in A$, as in standard bisimulation.

\begin{definition}[Collective bisimulation, simplicial]
Let simplicial models $\C = (C,\chi,\ell)$ and $\C' = (C',\chi',\ell')$ be given. A nonempty relation $\RR$ between $\FF(C)$ and $\FF(C')$ is a (collective) \emph{bisimulation} between $\C$ and $\C'$, notation $\RR: \C \bisim \C'$, iff for all $Y \in \FF(C)$ and $Y' \in \FF(C')$ with $(Y,Y') \in \RR$: 
\begin{itemize}
\item {\bf atoms}: $\ell(Y)=\ell'(Y')$. 
\item {\bf forth}: for all nonempty $B \subseteq A$, if $Z \in \FF(C)$ and $B \subseteq \chi(Y \inter Z)$ there is a $Z' \in \FF(C')$ with $B \subseteq \chi'(Y' \inter Z')$ such that $(Z,Z') \in \RR$.
\item {\bf back}: for all nonempty $B \subseteq A$, if $Z' \in \FF(C')$ and $B \subseteq \chi'(Y' \inter Z')$ there is a $Z \in \FF(C)$ with $B \subseteq \chi(Y \inter Z)$ such that $(Z,Z') \in\RR$.
\end{itemize}
If there is a bisimulation between $\C$ and $\C'$ we write $\C \bisim \C'$ and if there is one  between $\C$ and $\C'$ containing $(X,X')$ we write $(\C,X) \bisim (\C',X')$. 
\end{definition}

Also on simplicial models, collective bisimulation is different from standard bisimulation. For standard bisimulation, replace all `$B \subseteq$' by `$a \in$' in {\bf forth} and {\bf back}. Our notion therefore generalizes that of \cite{ledent:2019,goubaultetal_postdali:2021}. 

\begin{example}[Bisimilar but not collectively bisimilar, revisited] \label{ex.bcb}
We recall Example~\ref{ex.biscolbis} (Figure~\ref{fig.bisnoncbis}) wherein agents $a,b$ are uncertain about the value of agent $c$. The same information is represented in Fig.~\ref{zxcvzxcv} as simplicial models $\C$ and $\C'$. As there are three agents, a state in an epistemic model corresponds to a triangle, and if two states are indistinguishable for an agent $a$ then the corresponding triangles intersect in an $a$-vertex.

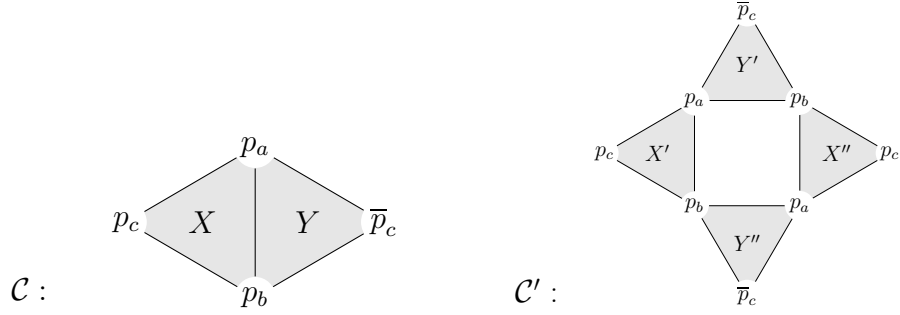
\begin{figure}[h]
\center
\begin{tikzpicture}[round/.style={circle,fill=white,inner sep=1}]
\fill[fill opacity = 0.1] (2,0) -- (2,2) -- (3.71,1) -- cycle;
\fill[fill opacity = 0.1] (0.29,1) -- (2,0) -- (2,2) -- cycle;
\node[round] (cc) at (-1,.1) {$\C:$};
\node[round] (b1) at (.29,1) {$p_c$};
\node[round] (b0) at (3.71,1) {$\np_c$};
\node[round] (c1) at (2,2) {$p_a$};
\node[round] (a0) at (2,0) {$p_b$};
\node (f0) at (1.3,1) {$X$};
\node (f2) at (2.7,1) {$Y$};
\draw[-] (b1) -- (a0);
\draw[-] (b1) -- (c1);
\draw[-] (a0) -- (b0);
\draw[-] (b0) -- (c1);
\draw[-] (a0) -- (c1);
\end{tikzpicture}
%
\hspace{1cm}
\scalebox{.7}{
\begin{tikzpicture}[round/.style={circle,fill=white,inner sep=1}]
\fill[fill opacity = 0.1] (4,0) -- (4,2) -- (5.71,1) -- cycle;
\fill[fill opacity = 0.1] (0.29,1) -- (2,0) -- (2,2) -- cycle;
\fill[fill opacity = 0.1] (2,2) -- (4,2) -- (3,3.71) -- cycle;
\fill[fill opacity = 0.1] (2,0) -- (4,0) -- (3,-1.71) -- cycle;
\node[round] (cc) at (-1,-1.61) {\Large $\C':$};
\node[round] (b1) at (.29,1) {$p_c$};
\node (f0) at (1.3,1) {$X'$};
\node (f0) at (4.7,1) {$X''$};
\node (f0) at (3,2.7) {$Y'$};
\node (f0) at (3,-.7) {$Y''$};
\node[round] (rb0) at (5.71,1) {$p_c$};
\node[round] (c1) at (2,2) {$p_a$};
\node[round] (a0) at (2,0) {$p_b$};
\node[round] (rc1) at (4,2) {$p_b$};
\node[round] (ra0) at (4,0) {$p_a$};
\draw[-] (b1) -- (a0);
\draw[-] (b1) -- (c1);
\draw[-] (ra0) -- (rb0);
\draw[-] (rb0) -- (rc1);
\draw[-] (a0) -- (c1);
\draw[-] (ra0) -- (rc1);
\node[round] (ac1) at (3,3.71) {$\np_c$};
\node[round] (bc1) at (3,-1.71) {$\np_c$};
\draw[-] (c1) -- (ac1);
\draw[-] (a0) -- (bc1);
\draw[-] (rc1) -- (ac1);
\draw[-] (ra0) -- (bc1);
\draw[-] (c1) -- (rc1);
\draw[-] (a0) -- (ra0);
\end{tikzpicture}
}
\caption{Simplicial models that are bisimilar but not collectively bisimilar}
\label{zxcvzxcv}
\end{figure}
These simplicial models are standardly bisimilar but not collectively bisimilar. The standard bisimulation is the relation $\RR = \{(X,X'),(X,X''),(Y,Y'),$ $(Y,Y'')\}$. Relation $\RR$ is not a collective bisimulation: 

Consider pair $(X,X') \in \RR$. Let us try to establish {\bf forth} for $\{a,b\}$. In $\C$, $\{a,b\} \subseteq X \inter Y$, but there is no facet $W$ in $\C'$ such that $\{a,b\} \subseteq X' \inter W$ and $(Y,W) \in \RR$. The only candidate would be $W=X'$, for which {\bf atoms} fails, as $Y$ and $X'$ have a different value for $p_c$.
%
%
\end{example}

\subsection{Semantics}

The semantics of the same language $\lang$ as in prior sections are now as follows. 
\begin{definition}[Semantics on simplicial models]
The interpretation of a formula $\phi\in \lang$ in a facet $X \in \FF(C)$ of a given simplicial model $\C = (C,\chi, \ell)$ is by induction on the structure of $\phi$. 
\[ \begin{array}{lcl}
\C,X \models p_a & \text{iff} & p_a \in \ell(X) \\
\C,X \models \neg\phi & \text{iff} & \C,X \not\models \phi \\
\C,X \models \phi\et\psi & \text{iff} & \C,X \models \phi \text{ and } \C,X \models \psi \\
\C,X \models D_B \phi & \text{iff} & \C,Y \models \phi \text{ for all } Y \in \FF(C) \text{ with } B \subseteq \chi(X \inter Y) \\
\C,X \models [\cp,\cg]\phi & \text{iff} & \C \oslash \cp, X_\cg \models \phi
\end{array} \]
Validity and modal equivalence on simplicial models are defined similarly to that on epistemic models. Formula $\phi$ is {\em valid} iff for all $(\C,X)$ we have that $\C,X \models \phi$; and given $(\C,X)$ and $(\C',X')$, by $(\C,X) \equiv (\C',X')$ we mean that for all $\phi \in \lang$: $\C,X \models \phi$ iff $\C',X' \models \phi$.
\end{definition} 

We now present examples of simplicial models, how to interpret distributed knowledge in them, how they informally correspond to epistemic models, and how to update them with communication patterns. Subsequently we show how epistemic models and simplicial models, and such updates, formally correspond. 

\begin{example}
In Fig.~\ref{zxcvzxcv} we have that $\C,X \not \models D_{ab} p_c$ whereas $
\C', X' \models D_{ab} p_c$. The latter holds because the only facet $W$ in $\C'$ with  $\{a,b\} \subseteq X' \inter W$ is $W=X'$, and $\C',X' \models p_c$.
\end{example}

\begin{example} \label{example.simpp}
We consider models for three agents $a,b,c$, (possibly) uncertain about atoms $p_a,p_b,p_c$, respectively. Fig.~\ref{nogeenfiguur} depicts the epistemic models and corresponding simplicial models. Names of some vertices and facets are explicit.

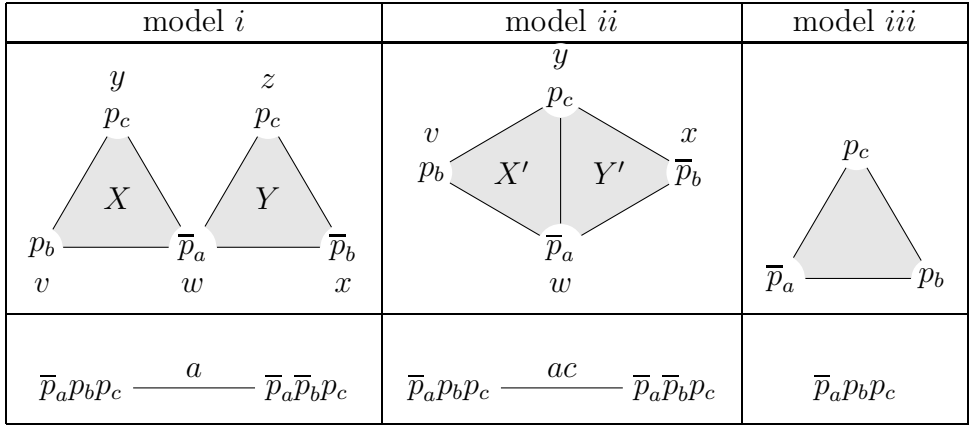
\begin{figure}[h]
\center
\begin{tabular}{|c|c|c|}
\hline
model $i$ & model $ii$ & model $iii$ \\
\hline
\begin{tikzpicture}[round/.style={circle,fill=white,inner sep=1}]
\fill[fill opacity = 0.1] (2,0) -- (4,0) -- (3,1.71) -- cycle;
\fill[fill opacity = 0.1] (0,0) -- (2,0) -- (1,1.71) -- cycle;
\node[round] (b1) at (0,0) {$p_b$};
\node[round] (b0) at (4,0) {$\np_b$};
\node[round] (c1) at (3,1.71) {$p_c$};
\node[round] (lc1) at (1,1.71) {$p_c$};
\node[round] (a0) at (2,0) {$\np_a$};
\node (ac1) at (3,2.21) {$z$};
\node (alc1) at (1,2.21) {$y$};
\node (bb1) at (0,-.5) {$v$};
\node (bb0) at (4,-.5) {$x$};
\node (ba0) at (2,-.5) {$w$};
\node (f0) at (3,.65) {$Y$};
\node (f1) at (1,.65) {$X$};
\draw[-] (b1) -- (a0);
\draw[-] (b1) -- (lc1);
\draw[-] (a0) -- (lc1);
\draw[-] (a0) -- (b0);
\draw[-] (b0) -- (c1);
\draw[-] (a0) -- (c1);
\end{tikzpicture}
&
\begin{tikzpicture}[round/.style={circle,fill=white,inner sep=1}]
\fill[fill opacity = 0.1] (2,0) -- (2,2) -- (3.71,1) -- cycle;
\fill[fill opacity = 0.1] (0.29,1) -- (2,0) -- (2,2) -- cycle;
\node[round] (b1) at (.29,1) {$p_b$};
\node[round] (b0) at (3.71,1) {$\np_b$};
\node[round] (c1) at (2,2) {$p_c$};
\node[round] (a0) at (2,0) {$\np_a$};
\node[round] (vb1) at (.29,1.5) {$v$};
\node[round] (vb0) at (3.71,1.5) {$x$};
\node[round] (vc1) at (2,2.5) {$y$};
\node[round] (va0) at (2,-.5) {$w$};
\node (f0) at (2.65,1) {$Y'$};
\node (f1) at (1.35,1) {$X'$};
\draw[-] (b1) -- (a0);
\draw[-] (b1) -- (c1);
\draw[-] (a0) -- (b0);
\draw[-] (b0) -- (c1);
\draw[-] (a0) -- (c1);
\end{tikzpicture}
&
\begin{tikzpicture}[round/.style={circle,fill=white,inner sep=1}]
\fill[fill opacity = 0.1] (2,0) -- (4,0) -- (3,1.71) -- cycle;
\node[round] (b0) at (4,0) {$p_b$};
\node[round] (c1) at (3,1.71) {$p_c$};
\node[round] (a0) at (2,0) {$\np_a$};
%
\draw[-] (a0) -- (b0);
\draw[-] (b0) -- (c1);
\draw[-] (a0) -- (c1);
\end{tikzpicture}
\\ 
\hline
&& \\
\begin{tikzpicture}
\node (010) at (.5,0) {$\np_ap_bp_c$};
\node (001) at (3.5,0) {$\np_a\np_bp_c$};
\draw[-] (010) -- node[above] {$a$} (001);
\end{tikzpicture}
&
\begin{tikzpicture}
\node (010) at (.5,0) {$\np_ap_bp_c$};
\node (001) at (3.5,0) {$\np_a\np_bp_c$};
\draw[-] (010) -- node[above] {$ac$} (001);
\end{tikzpicture}
&
\begin{tikzpicture}
\node (010) at (.5,0) {$\np_ap_bp_c$};
\end{tikzpicture} \\
\hline
\end{tabular}
\caption{Simplicial models and corresponding epistemic models}
\label{nogeenfiguur}
\end{figure}
First consider the public announcement by all of all they know, the communication pattern consisting of communication graph $U$ (the universal relation $\{a,b,c\}^2$). In the  simplicial model $(i)$ of Fig.~\ref{nogeenfiguur} we thus obtain the updated model on the left in Fig.~\ref{asdfasdf}. It consists of two disconnected parts (to avoid ambiguity we have framed the model in the figure).

Crucially, the vertex $w$ coloured with agent $a$ has been duplicated into two vertices in the updated simplicial model, $(w,X)$ and $(w,Y)$. This is because $\chi(Y)= \chi(X) = \{a,b,c\} = U a = U \chi(w)$. 

\begin{figure}[h]
\begin{center}
\begin{tabular}{|c|c|}
\hline
\begin{tikzpicture}[round/.style={circle,fill=white,inner sep=1}]
\fill[fill opacity = 0.1] (3,0) -- (5,0) -- (4,1.71) -- cycle;
\fill[fill opacity = 0.1] (0,0) -- (2,0) -- (1,1.71) -- cycle;
\node[round] (b1) at (0,0) {$p_b$};
\node[round] (b0) at (5,0) {$\np_b$};
\node[round] (c1) at (4,1.71) {$p_c$};
\node[round] (lc1) at (1,1.71) {$p_c$};
\node[round] (a0) at (2,0) {$\np_a$};
\node[round] (ra0) at (3,0) {$\np_a$};
\node (ac1) at (4,2.21) {$(z,Y)$};
\node (alc1) at (1,2.21) {$(y,X)$};
\node (bb1) at (0,-.5) {$(v,X)$};
\node (bb0) at (5,-.5) {$(x,Y)$};
\node (ba0) at (1.8,-.5) {$(w,X)$};
\node (rba0) at (3.2,-.5) {$(w,Y)$};
\node (f0) at (1,.65) {$X_U$};
\node (f1) at (5,.65) {$Y_U$};
\draw[-] (b1) -- (a0);
\draw[-] (b1) -- (lc1);
\draw[-] (a0) -- (lc1);
\draw[-] (ra0) -- (b0);
\draw[-] (b0) -- (c1);
\draw[-] (ra0) -- (c1);
\end{tikzpicture}
&
\begin{tikzpicture}[round/.style={circle,fill=white,inner sep=1}]
\fill[fill opacity = 0.1] (2,0) -- (4,0) -- (3,1.71) -- cycle;
\fill[fill opacity = 0.1] (0,0) -- (2,0) -- (1,1.71) -- cycle;
\node[round] (b1) at (0,0) {$p_b$};
\node[round] (b0) at (4,0) {$\np_b$};
\node[round] (c1) at (3,1.71) {$p_c$};
\node[round] (lc1) at (1,1.71) {$p_c$};
\node[round] (a0) at (2,0) {$\np_a$};
\node (ac1) at (3,2.21) {$(y,\{x,y\})$};
\node (alc1) at (1,2.21) {$(y,\{v,y\})$};
\node (bb1) at (0,-.5) {$(v,\{v\})$};
\node (bb0) at (4,-.5) {$(x,\{x\})$};
\node (ba0) at (2,-.5) {$(w,\{w\})$};
\node (f0) at (3,.65) {$Y_\cg$};
\node (f1) at (1,.65) {$X_\cg$};
\draw[-] (b1) -- (a0);
\draw[-] (b1) -- (lc1);
\draw[-] (a0) -- (lc1);
\draw[-] (a0) -- (b0);
\draw[-] (b0) -- (c1);
\draw[-] (a0) -- (c1);
\end{tikzpicture} \\
\hline
\end{tabular}
\end{center}
\caption{Left, the updated model $(i)$ of Fig.~\ref{nogeenfiguur}. Right, the updated model $(ii)$ of Fig.~\ref{nogeenfiguur}.}
\label{asdfasdf}
\end{figure}
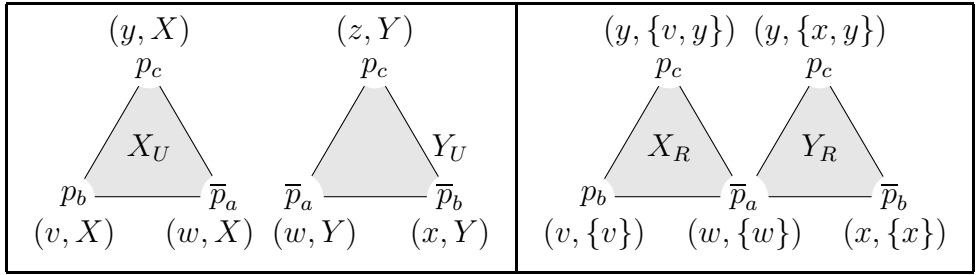

Now consider the simplicial model $(ii)$ in Fig.~\ref{nogeenfiguur} and the (singleton) communication pattern consisting of the communication graph $\cg^{bc} = \{(b,c)\} \union I$ wherein agent $c$ receives a message from agent $b$ (cf.\ Example~\ref{example.byzantine}). Although $a$ and $c$ were uncertain about the value of $p_b$, only $c$ will now learn that value. Fig.~\ref{asdfasdf} depicts the updated simplicial model. (Note that it is bisimilar ---even isomorphic--- to simplicial model $(i)$ in Fig.~\ref{nogeenfiguur}.)

Crucially, the vertex $w$ coloured with agent $a$ has now {\bf not} been duplicated into two vertices in the updated simplicial model: because $\cg^{bc} a = \{a\}$, only $(w,\{w\})$ is found in the updated simplicial model. Whereas agent $c$ has gained information from this update, as there are two different simplices coloured with $\cg^{bc} c = \{b,c\}$, namely $\{x,y\}$ and $\{v,y\}$.
\end{example}

\subsection{Correspondence between epistemic and simplicial models}

We continue by defining correspondence between epistemic models and simplicial models, and their updates with communication patterns. We first recall from \cite{GoubaultLR21,ledent:2019} the map $\sigma: \mathcal K \imp \mathcal S$ transforming an epistemic model into a simplicial model with the same information content, and the map $\kappa: \mathcal S \imp \mathcal K$ transforming a simplicial model into an epistemic model with the same information content, where we follow the presentation as in \cite{hvdetal.simpl:2022}.

\begin{definition}[Epistemic model to simplicial model]
Given an epistemic model $M = (W,\sim,L)$, we define $\sigma(M) = (C,\chi,\ell)$ as follows. The complex $C := \{ \{([w]_a,a) \mid a \in B\} \mid w \in W, B \subseteq A, B \neq \emptyset \}$, where we recall that $[w]_a = \{ v \in W \mid w \sim_a v\}$. Its vertices are therefore $\VV(C) = \{ ([w]_a,a) \mid a \in A, w \in W\}$, where we also define that $\chi([w]_a,a):=a$, and for all $p_a \in P$, $p_a \in \ell([w]_a,a)$ iff $p_a \in L(w)$. And its facets are therefore $\FF(C) = \{ \{([w]_a,a) \mid a \in A\} \mid w \in W\}$. For a facet we also write $\sigma(w)$. Then, $([w]_a,a) \in \sigma(w) \inter \sigma(v)$ iff $w \sim_a v$; that is, such that $[w]_a = [v]_a$. 
\end{definition}
This definition causes states that are indistinguishable for all agents to determine the same facet: if $[w]_a = [v]_a$ for all $a \in A$, then $\{([w]_a,a) \mid a \in A\} =  \{([v]_a,a) \mid a \in A\}$. For a pointed simplicial model $(\sigma(M),\sigma(w))$ we also write $\sigma(M,w)$.

\begin{definition}[Simplicial model to epistemic model]
Given a simplicial model $\C = (C,\chi,\ell)$, we define  an epistemic model  $\kappa(\C) = (W,\sim,L)$ as follows. Its domain $W$ is $\FF(C)$. For $X \in \FF(C)$ we may write $\kappa(X)$ instead of $X$ to make explicit that $X$ is a world and not a facet. We define $\kappa(X) \sim_a \kappa(Y)$ iff $a \in \chi(X \inter Y)$, and $L(\kappa(X)) = \ell(X)$.
\end{definition}
For a pointed epistemic model $(\kappa(\C),\kappa(X))$ (or $(\kappa(\C),X)$) we may write $\kappa(\C,X)$. It is easy to see that $\sigma(M)$ above is indeed a simplicial model, and that $\kappa(\C)$ above is indeed an epistemic model. In particular, $\kappa(\C)$ satisfies the requirement that it is local: if $X \sim_a Y$ in $\kappa(\C)$, there is a $v \in X \inter Y$ with $\chi(v)=a$; therefore, for all $p_a \in P_a$, $p_a \in \ell(v)$ iff $p_a \in \ell(X)$ and also $p_a \in \ell(v)$ iff $p_a \in \ell(Y)$, so that $p_a \in \ell(X)$ iff $p_a \in \ell(Y)$.

We continue by showing how this correspondence between epistemic models and simplicial models helps us to obtain results for bisimulation, update, truth, and validity.

We first show that our primitive dynamic object, the communication pattern, indeed similarly updates epistemic models and simplicial models; in other words, that the rather different update mechanisms after all produce `the same'  result. The properties shown in the following Th.~\ref{theo.zxcv} are visualized in the diagrams of Fig.~\ref{fig.zxcv}.

\begin{figure}[h]
\center
\begin{tikzpicture}
\node (00) at (0,0) {$M$};
\node (01) at (0,2) {$\sigma(M)$};
\node (10) at (2.5,0) {$M \odot \cp$};
\node (11) at (2.5,2) {$\sigma(M) \oslash \cp$};
\draw[->] (00) -- node[above,fill=white,inner sep=1pt] {$\odot$} (10);
\draw[->] (01) -- node[above,fill=white,inner sep=1pt] {$\oslash$} (11);
\draw[->] (00) -- node[right,fill=white,inner sep=1pt] {$\sigma$} (01);
\draw[->] (10) -- node[right,fill=white,inner sep=1pt] {$\sigma$} (11);
\end{tikzpicture}
\qquad\qquad
\begin{tikzpicture}
\node (00) at (0,0) {$\C$};
\node (01) at (0,2) {$\kappa(\C)$};
\node (10) at (2.5,0) {$\C \oslash \cp$};
\node (11) at (2.5,2) {$\kappa(\C) \odot \cp$};
\draw[->] (00) -- node[above,fill=white,inner sep=1pt] {$\oslash$} (10);
\draw[->] (01) -- node[above,fill=white,inner sep=1pt] {$\odot$} (11);
\draw[->] (00) -- node[right,fill=white,inner sep=1pt] {$\kappa$} (01);
\draw[->] (10) -- node[right,fill=white,inner sep=1pt] {$\kappa$} (11);
\end{tikzpicture}
\caption{Commuting diagrams visualizing Th.~\ref{theo.zxcv}}
\label{fig.zxcv}
\end{figure}

\begin{theorem} \label{theo.zxcv}
Let $M$, $\C$ and $\cp$ be given. Then $\sigma(M \odot \cp)$ is bisimilar to $\sigma(M)\oslash\cp$, and $\kappa(\C\oslash\cp)$ is bisimilar to $\kappa(\C)\odot\cp$.
\end{theorem}

\begin{proof}
To prove the first, consider the relation $\RR$ between the domain of $\sigma(M \odot \cp)$ and the domain of $\sigma(M)\oslash\cp$ defined as $\RR: \sigma(w,\cg) \mapsto \sigma(w)_\cg$ 
which induces a relation (with the same name) between vertices of these facets such that
 $\RR: ([w,\cg]_a,a) \mapsto  (([w]_a,a), \{([w]_b,b) \mid b \in \cg a\})$. This relation is a bisimulation:

\medskip

{\bf atoms}: $p_a \in \ell([w,\cg]_a,a)$, iff (def.\ of $\sigma$) $p_a \in \dot L(w,\cg)$, iff (def.\ of $\odot$) $p_a \in L(w)$, iff (def.\ of $\sigma$) $p_a \in \ell([w]_a,a)$, iff (def.\ of $\oslash$) $p_a \in \ddot\ell(([w]_a,a),\{([w]_b,b) \mid b \in \cg a\})$.

\medskip

{\bf forth}:  Given $(\sigma(w,\cg),\sigma(w)_\cg)\in\RR$, let $B \subseteq \sigma(w,\cg)\inter\sigma(v,S)$ be nonempty. As $\sigma(v)_S$ is by definition the required facet bisimilar to $\sigma(v,S)$, it is sufficient to show that $B \subseteq \chi(\sigma(w,\cg)\inter\sigma(v,S))$ implies $B \subseteq \chi(\sigma(w)_\cg \inter \sigma(v)_S)$. Instead of `implies' we show `if and only if', as this then also takes care of {\bf back}:

\bigskip

\noindent  
$B \subseteq \chi(\sigma(w,\cg)\inter\sigma(v,S))$\\ $\Eq$ \\ 
for all $a \in B$, $a \in  \chi(\sigma(w,\cg)\inter\sigma(v,S))$\\ $\Eq$ \\ 
for all $a \in B$, $([w,\cg]_a,a) = ([v,S]_a,a)$\\ $\Eq$ \\ 
for all $a \in B$, $(w,\cg) \dot\sim_a (v,S)$\\ $\Eq$ \\ 
for all $a \in B$, $w \sim_{\cg a} v$ and $\cg a = S a$   \\ $\Eq$ $(*)$ \\
for all $a \in B$, $([w]_a,a) = ([v]_a,a)$ \& $\{([w]_b,b) \mid b \in \cg a\}=\{([v]_b,b) \mid b \in S a\}$\\ $\Eq$ \\ 
for all $a \in B$, $(([w]_a,a),\{([w]_b,b) \mid b \in \cg a\}) = (([v]_a,a),\{([v]_b,b) \mid b \in S a\})$\\ $\Eq$ \\ 
for all $a \in B$, $a \in \chi(\sigma(w)_\cg \inter \sigma(v)_S)$\\ $\Eq$ \\ 
$B \subseteq \chi(\sigma(w)_\cg \inter \sigma(v)_S)$.

\bigskip 
\noindent $(*)$: $w \sim_{\cg a} v$, iff for all $b \in \cg a$, $w \sim_b v$, iff for all $b \in \cg a$, $[w]_b = [v]_b$, iff for all $b \in \cg a$, $([w]_b,b) = ([v]_b,v)$, iff $\{([w]_b,b) \mid b \in \cg a\}=\{([v]_b,b) \mid b \in S a\}$. Also note that $a \in \cg a$, so that in particular $([w]_a,a) = ([v]_a,a)$.

\medskip

{\bf back}: Similar to {\bf forth}.

\bigskip

To prove the second, consider the relation $Z$ between the domain of $\kappa(\C\oslash\cp)$ and the domain of $\kappa(\C)\odot\cp$ defined as $Z: X_\cg \mapsto (X,\cg)$. Also this relation is a bisimulation:
 
 \medskip
 
{\bf atoms}: $p_a \in L(X_\cg)$, iff $p_a \in \ddot\ell(X_\cg)$, iff $p_a \in \ddot\ell(v,Y)$ for $(v,Y)\in X_\cg$ with $\chi(v)=a$, iff (def.\ $\oslash$) $p_a \in \ell(v)$ with $v \in X$ and $\chi(v)=a$, iff  $p_a \in L(X)$, iff (def.\ $\odot$) $p_a \in \dot L(X,R)$.

\medskip

{\bf forth}:  As it is obvious that $Y_S$ is bisimilar to $(Y,S)$, similarly to above it is sufficient to show that $X_\cg \sim_B Y_S$ iff $(X,\cg) \sim_B (Y,S)$, that we can then use to prove not only {\bf forth} but also {\bf back}.

\bigskip

\noindent
$X_\cg \dot\sim_B Y_S$ \\ $\Eq$ \\
for all $a \in B$, $X_\cg \dot\sim_a Y_S$ \\ $\Eq$ (def.\ $\kappa$) \\
for all $a \in B$, $a \in \ddot\chi(X_\cg \inter Y_S)$\\ $\Eq$ (def.\ $\oslash$) \\
for all $a \in B$, $(v,Z) = (w,W)$ for $(v,Z)\in X_\cg$, $(w,W)\in Y_S$ with $\chi(v)=\chi(w)=a$\\ $\Eq$ $(**)$ \\ 
for all $a \in B$, $\cg a \subseteq \chi(X \inter Y)$ and $\cg a = S a$ \\ $\Eq$ (def.\ $\sigma$) \\
for all $a \in B$, $X \sim_{\cg a} Y$ and $\cg a = S a$ \\ $\Eq$ (def.\ $\dot\sim$) \\
for all $a \in B$, $(X,\cg) \dot\sim_a (Y,S)$ \\ $\Eq$ (def.\ $\odot$) \\
$(X,\cg) \sim_B (Y,S)$

\bigskip

\noindent $(**)$  Left to right: $(v,Z)=(w,W)$, iff $v =z$ and $Z = W$. Also, $Z=W$, $Z\subseteq X$, and $W \subseteq Y$ imply $Z \subseteq X \inter Y$. Also, as $(v,Z) \in X_\cg$,  $\chi(Z)= \cg a$ so that $\cg a \subseteq \chi(X \inter Y)$. Right to left: obvious.

\bigskip

{\bf back}: Similar to {\bf forth}.
\end{proof}

For standard bisimilarity between simplicial models it has been shown that for any $M$ and $\C$, $\kappa(\sigma(M)) \bisim M$ and $\sigma(\kappa(\C)) \bisim \C$ \cite{ledent:2019,GoubaultLR21}. The authors actually demonstrate isomorphy, based on an additional assumption that epistemic models are what they call \emph{proper}, which means that ${\inter_{a \in A}\sim_a} = {I}$. It is not hard to show that these results also hold for the collective bisimilarity of this contribution. Isomorphy does not hold.

\begin{proposition} \label{prop.fiets}
For any $M$ and $\C$, $\kappa(\sigma(M)) \bisim M$ and $\sigma(\kappa(\C)) \bisim \C$. 
\end{proposition}
\begin{proof}
Consider the relation $Z$ between $M$ and $\kappa(\sigma(M))$ defined as: for all $w \in W$, $Z: w \mapsto \{ ([w]_a,a) \mid a \in A \}$. We show that the relation $Z$ is a collective bisimulation.

\medskip

{\bf atoms}: $p_b \in L(w)$, iff (def.~$\sigma$) $p_b \in \ell(\{ ([w]_a,a) \mid a \in A \})$, iff (def.~$\kappa$, where we recall that the facets become the states) $p_b \in L(\{ ([w]_a,a) \mid a \in A \})$, which is $L(\kappa(\sigma(w))$.

\medskip

{\bf forth}: Let $w \sim_B v$. We show that $\{ ([v]_a,a) \mid a \in A \}$ satisfies the requirement. Obviously, $Z: v \mapsto \{ ([v]_a,a) \mid a \in A \}$. It remains to show that 
$\{ ([w]_a,a) \mid a \in A \} \sim_B \{ ([v]_a,a) \mid a \in A \}$:

\bigskip

\noindent
$w \sim_B v$ \\
$\Eq$ \\ 
$w \sim_b v$ for all $b \in B$ \\
$\Eq$ (def.~$\sigma$) \\
$([w]_b,b) \in \{ ([w]_a,a) \mid a \in A \} \inter \{ ([v]_a,a) \mid a \in A \}$ for all $b \in B$ \\ 
$\Eq$ (def.~of the chromatic map $\chi$ on $\sigma(M)$) \\
$b \in \chi(\{ ([w]_a,a) \mid a \in A \} \inter \{ ([w]_a,a) \mid a \in A \})$ for all $b \in B$ \\ $\Eq$ (def.~$\sim_b$ on $\kappa(\sigma(M))$) \\
$\{ ([w]_a,a) \mid a \in A \} \sim_b \{ ([w]_a,a) \mid a \in A \})$ for all $b \in B$ \\
$\Eq$ \\
$\{ ([w]_a,a) \mid a \in A \} \sim_B \{ ([v]_a,a) \mid a \in A \}$.

\bigskip

{\bf back}: This is shown similarly to {\bf forth}.

\bigskip

Now consider the relation $\RR$ between $\C$ and $\sigma(\kappa(\C))$ defined as: for all facets $X$ in $\C$, $\RR: X \mapsto \{([X]_a,a) \mid a \in A\}$. We show that the relation $\RR$ is a collective bisimulation.

\bigskip

{\bf atoms}: $p_b \in \ell(X)$, iff (def.~$\kappa$) $p_b \in L(X)$ in $\kappa(\C)$, iff (def.~$\sigma$) $p_b \in \ell([X]_b,b)$, iff (def.\ of $\ell$ on faces)  $p_b \in \ell(\{([X]_a,a) \mid a \in A\}) = \ell(\sigma(\kappa(X))$.

\bigskip

{\bf forth}: Let $B \subseteq \chi(X \inter Y)$ for a facet $Y$ of $\C$. Choose $\RR: Y \mapsto \{([Y]_a,a) \mid a \in A\}$. We need to show that $B \subseteq \chi(\{([X]_a,a) \mid a \in A\} \inter \{([Y]_a,a) \mid a \in A\})$:

\bigskip

\noindent
$B \subseteq \chi(X \inter Y)$ \\
$\Eq$ \\
$b \in \chi(X \inter Y)$ for all $b \in B$ \\
$\Eq$ (def.~$\kappa$) \\
$X \sim_b Y$ for all $b \in B$ \\
$\Eq$ (def.~$\sigma$) \\
$([X]_b,b) \in \{([X]_a,a) \mid a \in A\} \inter \{([Y]_a,a) \mid a \in A\}$ for all $b \in B$ \\
$\Eq$ \\
$b \in \chi(\{([X]_a,a) \mid a \in A\} \inter \{([Y]_a,a) \mid a \in A\})$ for all $b \in B$ \\
$\Eq$ \\
$B \subseteq \chi(\{([X]_a,a) \mid a \in A\} \inter \{([Y]_a,a) \mid a \in A\})$.

\bigskip

{\bf back}: This is shown similarly to {\bf forth}.
\end{proof}
We continue by showing that, also on simplicial models, collective bisimilarity implies modal equivalence  in the language $\lang$. 

\begin{theorem} \label{theo.final}
$(\C,X) \bisim (\C',X')$ implies $(\C,X) \equiv (\C',X')$.
\end{theorem}

\begin{proof}
We prove by induction on $\phi$ that ``for all $\phi\in\lang$, for all pointed simplicial models $(\C,X)$, $(\C',X')$: if $(\C,X) \bisim (\C',X')$ then $\C,X\models\phi$ iff $\C',X'\models\phi$''. The non-standard cases are the distributed knowledge case $D_B\phi$ and the update case $[\cp,\cg]\phi$, to which we therefore restrict the proof.

{\bf Case knowledge.} \ Given $(\C,X)$, $(\C',X')$, and $(\C,X)\bisim(\C',X')$, assume $\C',X' \models D_B\phi$. Let $\RR$ be a bisimulation between $\C$ and $\C'$ containing $(X,X')$. Let $Y \in C$ be arbitrary such that $B \subseteq \chi(X \inter Y)$. From $B \subseteq \chi(X \inter Y)$ and {\bf forth} it follows that there is a $Y' \in C'$ with $B \subseteq \chi'(X' \inter Y')$ and $(Y,Y') \in \RR$. From $B \subseteq \chi'(X' \inter Y')$ and $\C',X' \models D_B\phi$ it follows (by the semantics of $D_B$) that $\C',Y' \models \phi$. From $\C',Y' \models \phi$ and by induction hypothesis it follows that $\C,Y \models \phi$. Therefore, as $Y$ was arbitrary, $\C,X \models D_B\phi$. The other direction is similar.

{\bf Case update.} \ Given $(\C,X)$, $(\C',X')$, and $(\C,X)\bisim(\C',X')$, assume $\C,X \models [\cp,\cg]\phi$. Then $\C \oslash \cp, X_\cg \models \phi$. Let $\RR$ be a bisimulation between $\C$ and $\C'$ containing $(X,X')$. We first show that $\RR'$  defined as: $(X_\cg,X'_{\cg'}) \in \RR'$ if $(X,X') \in \RR$ and $\cg=\cg'$ 
is a bisimulation between $(\C \oslash \cp, X_\cg)$ and $(\C' \oslash \cp, X'_{\cg})$ containing $(X_\cg,X'_{\cg})$.

\medskip

{\bf atoms}: $\ddot \ell(X_\cg)=\ell(X)= \ell'(X')= \ddot \ell'(X'_\cg)$.

\medskip

{\bf forth}:  Let $B \subseteq X_\cg \inter Y_S$. As $(X_\cg,X'_{\cg}) \in \RR'$, also $(X,X') \in \RR$. As $B \subseteq X_\cg \inter Y_S$, also $B \subseteq X \inter Y$. From forth for $\RR$ we obtain a $Y'$ such that $(Y,Y') \in \RR$ and $B \subseteq X' \inter Y'$. We now show that $Y'_S$ is the required domain element to close {\bf forth}. For this, it remains to prove that $B \subseteq X_\cg \inter Y_S$ iff $B \subseteq X'_\cg \inter Y'_S$. Let $B \subseteq X_\cg \inter Y_S$ and $a \in B$. Then $(v,W) \in X_\cg$ with $\chi(v)=a$ satisfies that $W \subseteq X$ and $\chi(W)=\cg a$, and as $(v,W)\in Y_S$ as well, also $W \subseteq Y$ and $\chi(W)= S a$. We now (crucially) use that $(X,X') \in \RR$ induces a pointwise relation between all $x \in X$ and $x' \in X'$ with $\chi(x)=\chi'(x')=b$ for any $b \in \cg a$ (where we note that $b$ may not be in $B$), and similarly for $(Y,Y') \in \RR$. Therefore, for any such $(v,W)\in X_\cg$, we obtain $(v',W') \in X'_\cg$, so that $(v,W)\in X_\cg \inter Y_S$ implies $(v',W') \in X'_\cg \inter Y'_S$, as required.

\medskip

{\bf back}: Similar to {\bf forth}.

\medskip

Having established that $(\C\oslash\cp, X_\cg)$ is bisimilar to $(\C'\oslash\cp,X'_\cg)$, we now apply the induction hypothesis for $\phi$ on pointed models $(\C\oslash\cp, X_\cg)$ and $(\C'\oslash\cp,X'_\cg)$ thus obtaining that $(\C'\oslash\cp,X'_\cg) \models \phi$, so that $\C',X' \models [\cp,\cg]\phi$, as required.
\end{proof}
Even for the restricted language $\lang^-$ this result of Th.~\ref{theo.final} is novel. One can also show that $(\C,X) \equiv (\C',X')$ implies $(\C,X) \bisim (\C',X')$ on so-called {\em star-finite} simplicial models (and complexes). A complex $C$ is {\em star-finite} if for all faces $Y \in C$, the {\em star} of $Y$, defined as $\{ Z \in C \mid Y \subseteq Z\}$, is a finite set. This direction is similar to the proof that modal equivalence implies bisimulation on epistemic models; in particular it suffices to consider formulas in $\lang^-$. It is therefore omitted.\footnote{Sketch: define relation $\RR: X \mapsto X'$ iff $(C,X) \equiv (C',X')$. We show that $\RR$ is a bisimulation. {\bf Atoms} is obvious. Concerning {\bf forth}, let be given $(X,X') \in \RR$ and $B \subseteq X \inter Y$, and the \emph{finite} set of $Y'_1,\dots,Y'_n$ such that $B \subseteq X' \inter Y'_1$. If none of those are bisimilar to $Y$, there are distinguishing formulas $\phi_1,\dots,\phi_n$ all true in $Y$ but false in $Y'_1,\dots,Y'_n$ resp. Then, $C,X \models \hat{D}_B (\phi_1 \et\dots\et\phi_n)$ whereas $C',X' \not\models \hat{D}_B (\phi_1 \et\dots\et\phi_n)$, contradicting $(C,X)\equiv(C',X')$.} We therefore have again the Hennessy-Milner property for this semantics.
 
Finally, we now show that the $\sigma$ and $\kappa$ mappings preserve truth (value). This implies that validity is the same for both semantics, and that there is therefore only one communication pattern logic.

\begin{proposition} \label{prop.sametruth}
Let $\phi \in \lang$, $(M,w)$, and $(\C,X)$ be given. Then:
\begin{enumerate}
\item $M,w \models \phi$ iff $\sigma(M,w) \models \phi$; \label{aa}
\item $\C,X \models \phi$ iff $\kappa(\C,X) \models \phi$. \label{bb}
\end{enumerate}
\end{proposition}
\begin{proof}
The proof is by induction on the structure of $\phi$. The Boolean cases are elementary and left to the reader.

\bigskip

\noindent
$M,w \models D_B\phi$, iff $M,v \models \phi$ for all $v \sim_B w$, iff (induction) $\sigma(M,v) \models \phi$ for all $v \sim_B w$, iff (def.~$\sigma$) $\sigma(M,v) \models \phi$ for all $\sigma(v)$ such that for all $b \in B$, $([v]_b,b) \in \sigma(v) \inter \sigma(w)$, iff $\sigma(M,v) \models \phi$ for all $\sigma(v)$ such that $B \subseteq \chi(\sigma(v) \inter \sigma(w))$, iff (semantics of distributed knowledge) $\sigma(M,w) \models D_B\phi$.

\bigskip

\noindent
$M,w \models [\cp,\cg]\phi$, iff (semantics on epistemic models) $M\odot\cp, (w,\cg) \models \phi$, iff (induction) $\sigma(M\odot\cp), (w,\cg) \models \phi$, iff (Th.~\ref{theo.zxcv}) $\sigma(M) \oslash\cp, w_\cg \models \phi$, iff (semantics on simplicial models) $\sigma(M), w \models [\cp,\cg]\phi$.

\bigskip

\noindent
$\C,X \models D_B\phi$, iff $\C,Y \models \phi$ for all $Y$ with $B \subseteq \chi(X \inter Y)$, iff (induction) $\kappa(\C),Y \models \phi$ for all $Y$ with $B \subseteq \chi(X \inter Y)$, iff (def.~$\kappa$) $\kappa(\C),Y \models \phi$ for all $Y$ with $X \sim_B Y$, iff $\kappa(\C),X \models D_B \phi$.

\bigskip

\noindent
$\C,X \models [\cp,\cg]\phi$, iff $\C\oslash\cp, X_\cg \models \phi$, iff (induction) $\kappa(\C\oslash\cp, X_\cg) \models \phi$, iff (Th.~\ref{theo.zxcv}) $\kappa(\C) \odot \cp, (X,\cg) \models \phi$, iff (semantics) $\kappa(\C),X \models [\cp,\cg]\phi$.
\end{proof}
Props.~\ref{prop.sametruth}.\ref{aa} and \ref{prop.sametruth}.\ref{bb} jointly establish that $M,w \models \phi$ iff $\kappa(\sigma(M,w)) \models \phi$ and also that $\C,X \models \phi$ iff $\sigma(\kappa(\C,X)) \models \phi$, and therefore:
\begin{corollary} For all $(M,w)$ and $(\C,X)$: $(M,w)\equiv\kappa(\sigma(M,w))$ and  $(\C,X)\equiv\sigma(\kappa(\C,X))$. \end{corollary} The first already followed from Prop.~\ref{prop.fiets} and Th.~\ref{theorem.xx}. The second already followed from Prop.~\ref{prop.fiets} and Th.~\ref{theo.final}.

\begin{theorem} \label{theo.samelogic}
The logic of communication patterns on epistemic models and the logic on communication patterns on simplicial models are the same.
\end{theorem}
\begin{proof}
With Prop.~\ref{prop.sametruth} we first show that a formula is valid on epistemic models iff it is valid on simplicial models. 

Let $\phi\in\lang$ be given and assume that $\phi$ is valid on epistemic models. Let now $(\C,X)$ be arbitrary. Then $\kappa(\C,X)$ is an epistemic model. As $\phi$ is valid, $\kappa(\C,X) \models \phi$. Using Prop.~\ref{prop.sametruth}.\ref{bb}, we obtain that $\C,X\models \phi$. Therefore $\phi$ is valid on simplicial models. 

The proof in the other direction is similar, only now featuring $(M,w)$ and $\sigma(M,w)$.

As the validities in the respective semantics are the same, and the axiomatization of communication pattern logic interpreted on epistemic models is sound and complete, this is therefore the same logic when interpreted on simplicial models. 
\end{proof}
Th.~\ref{theo.samelogic} finally justifies why we did not distinguish notation for validity in the different semantics: $\models \phi$ means that $\phi$ is valid either way.

As the logics are the same for the different semantics, in particular the axiomatization is also a reduction system for the simplicial semantics, allowing to eliminate communication pattern modalities. We could therefore also have obtained Th.~\ref{theo.final} by restricting that proof to modal equivalence in the language $\lang^-$. However, as mentioned, this is already a result. Th.~\ref{theo.final} and Th.~\ref{theorem.xx} seem to contrast well.


\section{Conclusion and further research} \label{sec.comparison}

\paragraph*{Conclusion} We presented communication pattern logic to reason about distributed systems wherein agents communicate to each other all they know. We provided an original update construction, an axiomatization, 
apart from a semantics on Kripke models also one on simplicial complexes, and we showed how these semantics correspond, also in relation to collective bisimulation for distributed knowledge.

\paragraph*{Distributed computing and dynamic epistemic logic}
The investigation of dynamic epistemic logics with dynamic modalities for communication patterns seems promising for applications in distributed computing. Communication patterns are surprisingly different from other update mechanisms such as action models (discussed below) and arrow updates \cite{kooirenne}. They are suitable to model multi-agent system communication where we wish to abstract from the content of messages. Their dual use as update mechanisms on simplicial complexes appears to widen the scope for their use: combinatorial topology is an established field in distributed computing. Communication patterns appear to suitably formalize computability of distributed tasks that are known to require high-dimensional topological arguments, like \emph{consensus} (work in progress continuing \cite{diego:2021}), \emph{equality negation} (where it appears promising to proceed as in \cite{goubaultetal_postdali:2021}), \emph{set agreement}, and \emph{renaming}.

Even apart from dynamics, research questions on the relation between epistemic and simplicial semantics abound. Various questions, in particular of philosophical logical relevance, were addressed in the recent \cite{hvdetal.simpl:2022}. How restrictive is the requirement that epistemic models be local (that agents know their local state) \cite[Sect.~9]{hvdetal.simpl:2022}? Only such models correspond to complexes, whereas complexes are local by definition, as their building stones are agent-coloured vertices. How to model systems where some processors/agents have crashed, or, in more familiar philosophical logical terms, how to conceive multi-agent Kripke models wherein some agents have died \cite{Ditmarsch21,GoubaultLR22}? Can dead agents know something? Such subjects have established roots in philosophical logic and AI \cite{Lambert1967-KARFLA,DworkM90}. How to employ common knowledge and {\em common distributed knowledge}, a novel epistemic notion of proposed in \cite{vanwijk:2015,Baltag20}, to describe topological primitives in complexes? Common distributed knowledge characterizes \emph{manifolds} in complexes \cite[Sect.~6.2]{hvdetal.simpl:2022}). How to model belief on simplicial complexes \cite[Sect.~7]{hvdetal.simpl:2022}, and Byzantine phenomena such as deception and error \cite{abs-2106-11499} (in a runs-and-systems setting)?

Given the currently growing community of researchers active in this area we hope that many such issues and open questions will be investigated in the coming years and receive answers, and that the investigation of the dynamics of such systems will get a boost from the notion of communication pattern.

\bigskip

Let us shortly also describe in more technical detail three areas for further research. 

\paragraph*{Update expressivity} We wish to compare the \emph{update expressivity} of communication pattern logic and \emph{action model logic} \cite{baltagetal:1998}. Communication patterns, like action models, are examples of {\em updates} transforming epistemic models into other epistemic models. Communication patterns can always be executed, but action models have preconditions for their execution. For example, a truthful public announcement of $p$ requires $p$ to be true, otherwise it would not be truthful. One can compare updates by their {\em update expressivity} \cite{jveetal:2012,kooirenne,hvdetal.aus:2020}. Given a class of pointed models, such as the pointed epistemic models $(M,w)$ given a set of agents $A$ and atoms $P$ in our contribution, each update determines a binary relation on this class.\footnote{One pair in that relation for the update `public announcement of $p$' would then be $((M,w),(M',w))$, where $M'$ is the restriction of $M$ to the $p$-states, and this on condition that $w$ is also in that restriction.} This relation is a partial order. 

It is fairly obvious that communication pattern logic is not at least as update expressive as action model logic. In a public announcement, the environment may reveal something that cannot be revealed by the agents jointly, such as in Ex.~\ref{ex.pa} the announcement of $m_a \vel m_b \vel m_c$, `One of you is muddy', by the father, who is not modelled as an agent. Also, even when an announcement is made by an agent in the system, she may choose to reveal only some but not all of her local variables, such as, if $a$ knows $p_a$ and $p'_a$, $a$ informing $b$ of $p_a$ but not of $p'_a$. 

It is not obvious, we think, that action model logic is not at least as update expressive as communication pattern logic, except for the case where there are infinitely many local variables, as an infinite conjunction of literals (variables or their negation) cannot be announced. However, if there are finitely many local variables, proving the same requires a delicate argument \cite{armandoetal.tark:2023}. 

There is also {\em expressivity} in the usual sense, namely that every formula in a given language is equivalent to a formula in another logical language. One can then simply observe that communication pattern logic and action model logic, both with distributed knowledge, are equally expressive, as they both reduce to the logic of distributed knowledge. The latter was shown in \cite[Th.~15]{WangA15}.

\paragraph*{How to represent a communication pattern?} Although we represent a communication pattern as a set of communication graphs, we can also see it as a kind of Kripke model, and also as a kind of simplicial complex. Such alternative representations may make it easier to conceive and obtain more theoretical results for our framework.

Given $\cp$, {\em communication pattern model} $(\cp,\sim)$ is the structure where for each $a \in A$, $\sim_a$ is an equivalence relation defined as: $\cg \sim_a \cg'$ iff $\cg a = \cg' a$. The update of an epistemic model with a communication pattern model is now a restricted  product. It is restricted because the relations are restricted. We recall that $(w,\cg) \sim_a (w',\cg')$ iff $w \sim_{\cg a} w'$ and $\cg a= \cg' a$. This now has become the requirement that $w \sim_{\cg a} w'$ and $\cg \sim_a \cg'$. This is the representation of communication patterns in \cite{diego:2021}, and a similar representation is chosen for the {\em arbitrary reading events} in \cite[Sect.~6]{Baltag20}.

Alternatively, given $\cp$, {\em simplicial communication pattern} $(C^{\cp}, \chi^{\cp})$ is the chromatic simplicial complex with simplices $\{ (a,\cg a) \mid a \in B, B \subseteq A, \cg \in \cp \}$, and for all $(a,\cg a) \in \VV(C^{\cp})$, $\chi^{\cp}(a,\cg a)=a$. The update of a simplicial model with a simplicial communication pattern is again some kind of product (but not a restricted product). The simplicial communication pattern is reminiscent of the simplicial action model in \cite{GoubaultLR21,hvdetal.simpl:2022} and in \cite[Def.\ 3.34 \& 3.35]{ledent:2019}. 

\begin{example}
Once more we recall Ex.~\ref{example.byzantine} modeling Byzantine attack, and communication pattern $\cp = \{I,R^{ab}\}$ where $R^{ab} = I \union \{(a,b)\} = \{(a,a),(b,b),(a,b)\}$. Let us show the different representations.
\begin{enumerate}
\item
communication pattern:  $\{I, R^{ab}\}$
\item
communication pattern model: $I \stackrel{b}{\text{------}}R^{ab}$
\item \label{cpp}
communication pattern complex:  $(b,I b)$------$(a,R^{ab} a)$------$(b,R^{ab} b)$ 
\item
item \ref{cpp} but with explicit groups: $(b,\{b\})$------$(a,\{a\})$------$(b,\{a,b\})$
\end{enumerate}
The first item shows once the communication pattern as a set. The second item is a corresponding communication pattern model (where ${\sim_a} = I$). The third item is a communication pattern complex. As it is for two agents it is of dimension $1$ (no triangles but edges, as in Ex.~\ref{annebill}). We recall that $I b = \{b\}$, $I a = R^{ab} a = \{a\}$, and $R^{ab} = \{a,b\}$ so that the communication pattern complex can more appealing be visualized as the one below it. 
\end{example}

\paragraph*{History-based epistemic models} With {\em history-based semantics} \cite{jfaketal.JPL:2009} we can keep track of rounds of communicative events. Thus we can make precise how \emph{full-information protocols}, where in each round all agents send the tree composed of their own call history and the histories they have received from other agents, correspond to iterated update with communication patterns. These matters are also investigated in \cite{diego:2021} and ongoing work continuing that. In \cite{armandoetal.tark:2023}, history-based semantics allow to compare action models and communication patterns on update expressivity. 

\subsection*{Acknowledgements} We thank the reviewers for their constructive comments, and for their encouragement to substantially expand the original submission. Armando Castañeda was partially supported by PAPIIT projects IN108720
and IN108723. Hans van Ditmarsch is grateful to Sergio Rajsbaum for opening his eyes to the beautiful duality between Kripke models and simplicial complexes, during Hans' visit to Mexico in 2019, and for introducing him to many of his collaborators and students, including the co-authors of this work. 


\bibliographystyle{plainurl}
\bibliography{biblio2023}

\end{document}